\documentclass[11pt,a4paper]{article}
\usepackage{amsmath,amssymb,amsthm,graphicx,braket,multirow,authblk,amsthm,dcolumn,latexsym,subcaption,mathtools,cite,cancel}
\usepackage[normalem]{ulem}
\usepackage[a4paper]{geometry}
\DeclareSymbolFontAlphabet{\amsmathbb}{AMSb}%

\usepackage[unicode=true,
bookmarks=true,bookmarksopen=false,
breaklinks=false,pdfborder={0 0 0},colorlinks=true]
{hyperref}
\usepackage{xcolor}
\definecolor{cblue}{rgb}{0.16, 0.32, 0.75}
\definecolor{cred}{rgb}{0.7, 0.11, 0.11}
\hypersetup{%
	,linkcolor=cred
	,citecolor=cblue
	,urlcolor=black
}
\usepackage[normalem]{ulem}

\def\<{\langle}
\def\>{\rangle}
\def\oper{{\mathchoice{\rm 1\mskip-4mu l}{\rm 1\mskip-4mu l}
		{\rm 1\mskip-4.5mu l}{\rm 1\mskip-5mu l}}}

\newcommand{\Dt}{\Delta t}
\newcommand{\tr}{\mathrm{Tr}}

\newcommand{\A}{\mathop{\mathcal{A}}\nolimits}

\newcommand{\ketbra}[2]{| #1 \rangle\!\langle #2 | }

\newtheorem{theorem}{Theorem}[section]
\newtheorem{proposition}{Proposition}[section]
\newtheorem{corollary}{Corollary}[section]

\newtheorem{lemma}{Lemma}[section]
\theoremstyle{remark}\newtheorem{remark}{Remark}[section]

\newcommand{\hilb}{\mathcal{H}}
\newcommand{\vac}{\mathrm{vac}}
\renewcommand{\i}{\mathrm{i}}
\newcommand{\e}{\mathrm{e}}
\newcommand{\g}{\mathrm{g}}
\renewcommand{\A}{\mathrm{A}}
\renewcommand{\L}{\mathrm{L}}
\renewcommand{\H}{\mathrm{H}}

\begin{document}	
	\title{
		\textbf{On Markovianity and classicality in multilevel spin--boson models}
	}
	
	\author[$\hspace{0cm}$]{Dariusz Chru\'sci\'nski$^{1,}$\footnote{darch@fizyka.umk.pl}}
	\author[$\hspace{0cm}$]{Samaneh Hesabi$^{1,}$\footnote{samaneh.hesabi@umk.pl}}
	\affil[$1$]{\small Institute of Physics, Faculty of Physics, Astronomy and Informatics, Nicolaus Copernicus University, Grudziadzka 5/7, 87-100 Toru\'n, Poland}
	
	\author[$\hspace{0cm}$]{Davide Lonigro$^{2,3,}$\footnote{davide.lonigro@ba.infn.it}}
	\affil[$2$]{\small Dipartimento di Fisica and MECENAS, Universit\`{a} di Bari, I-70126 Bari, Italy}
	\affil[$3$]{\small INFN, Sezione di Bari, I-70126 Bari, Italy}
	
	\maketitle
	\vspace{-0.5cm}	
	
	\begin{abstract}
		We provide a detailed discussion about the unitary and reduced evolution induced by family of Hamiltonian models describing a multilevel system, with a ground state and a possibly multilevel excited sector, coupled to a multimode boson field via a rotating-wave interaction. We prove explicitly that the system, in the limit in which the coupling is flat with respect to the boson frequencies, is Markovian under sharp measurements in arbitrary bases; we also find necessary and sufficient conditions under which the process is classical, i.e. its family of multitime joint probability distributions satisfies the Kolmogorov consistency condition, and may thus be equivalently obtained by a classical stochastic process.
	\end{abstract}
	
	\maketitle
	
	\section{Introduction}
	
	Long before the dawn of quantum mechanics, randomness has always occupied a central role in the description of natural phenomena. Since the publication of A. Kolmogorov's pioneering book~\cite{Kolmogorov}, the modern theory of stochastic processes has been developed in a measure-theoretical axiomatic framework which proves to be powerful and flexible enough for the description of diverse classical physical phenomena---Brownian motion being a paradigmatic example---as well as having applications in other natural and social sciences~\cite{Kampen}.
	
	In this framework, a (classical) stochastic process corresponds to a collection of functions $\{X(t)\}_{t\in T}$, with $T\subset\mathbb{R}$ being the set of times, from a suitably chosen probability space to some set of values $\mathcal{X}$. Most importantly, it is associated with a family of \textit{joint probability distributions} $\{\mathbb{P}_n\}_{n\in\mathbb{N}}$ defined as follows: given a family of values $x_1,\dots,x_n\in\mathcal{X}$ and a set of times $t_1,\dots,t_n$, the quantity
	\begin{equation}
		\mathbb{P}_n(x_n,t_n;x_{n-1},t_{n-1};\ldots;x_1,t_1)=\mathrm{Prob}\left\{X(t_1)=x_1,\dots,X(t_n)=x_n\right\}
	\end{equation}
	represents the probability that the system, probed at each time $t_j$, is measured in the state $x_j$. Interestingly, the distributions associated with a stochastic process automatically satisfy a \textit{consistency condition} which, for a discrete set of values, reads as follows:
	\begin{equation}\label{eq:consistency}
		\mathbb{P}_{n-1}\left(x_n,t_n;\dots;\cancel{x_j,t_j};\dots;x_1,t_1\right)=\sum_{x_j\in\mathcal{X}}\mathbb{P}_n(x_n,t_n;\dots;x_j,t_j;\dots;x_1,t_1),
	\end{equation}
	that is, for all $k<n$ all $k$-time joint distributions can be obtained via marginalization from the $n$-time one. Conversely, the celebrated Kolmogorov extension theorem guarantees that every \textit{consistent} (i.e. satisfying Eq.~\eqref{eq:consistency}) family of joint probability distributions can be obtained from a classical stochastic process. Importantly, as pointed out in~\cite{Andrea-1,Andrea-2}, Eq.~\eqref{eq:consistency} essentially means that not performing a measurement at the time $t_j$ is operationally indistinguishable from performing the measurement at the time $t_j$ and then ``forgetting'' about the outcome, i.e. averaging over all possible outcomes---that is, \textit{measurements do not alter the state of the system}.
	
	Markovian processes are particularly important. A stochastic process is said to be \textit{Markovian}\footnote{
		While this definition is apparently asymmetric in time, equivalent definitions of Markovianity which restore the symmetry between “future" and “past" can be given; see e.g.~\cite[Definition 7.1]{cufaro}. Essentially, Markovianity means that the past and the future are conditionally mutually independent with respect to the present time.
	} whenever the following property holds for all $x_1,\dots,x_n\in\mathcal{X}$ and $t_1\leq \ldots\leq t_n\in T$:
	\begin{equation}\label{eq:markov}
		\mathbb{P}(x_n,t_n|x_{n-1},t_{n-1};\dots;x_1,t_1)=\mathbb{P}(x_n,t_n|x_{n-1},t_{n-1}),
	\end{equation}
	with $\mathbb{P}(\cdot|\cdot)$ being the conditional probability; in words, the process is Markovian if the information contained in the more recent measure is fully representative of the past history of the system. In this sense, a Markovian process is memoryless. Most importantly, as a straightforward consequence of Eq.~\eqref{eq:markov}, the whole family of joint probabilities can be entirely reconstructed from the single-time probability $\mathbb{P}_1(x,t)$ and the transition probability $\mathbb{P}(x,t|y,s)$ via
	\begin{equation}\label{eq:reconstruct}
		\mathbb{P}_n(x_n,t_n;x_{n-1},t_{n-1};\ldots;x_1,t_1)=\mathbb{P}(x_n,t_n|x_{n-1},t_{n-1})\cdots\mathbb{P}(x_2,t_2|x_{1},t_{1})\mathbb{P}_1(x_1,t_1);
	\end{equation}
	besides, by Eqs.~\eqref{eq:markov}--\eqref{eq:reconstruct}, the whole family of consistency conditions~\eqref{eq:consistency} reduces to the Chapman--Kolmogorov equations~\cite{Kampen}:
	\begin{eqnarray}\label{eq:chapman0}
		\mathbb{P}_1(x,t)&=&\sum_{z\in\mathcal{X}}\mathbb{P}(x,t|z,r)\mathbb{P}_1(z,r),\qquad\quad\,\, t\geq r\\
		\label{eq:chapman}
		\mathbb{P}(x,t|y,s)&=&\sum_{z\in\mathcal{X}}\mathbb{P}(x,t|z,r)\mathbb{P}(z,r|y,s),\qquad t\geq r\geq s.
	\end{eqnarray}
	
	Quantum mechanics is, in its essence, a probabilistic theory: the outcome of any quantum mechanical experiment is famously non-deterministic. However, this is far from being the only novelty of quantum mechanics: the plethora of new phenomena unveiled by a close scrutiny to the quantum world, like entanglement and decoherence, makes the language of classical stochastic processes unsuited to quantum mechanics. To put it in mathematical terms, let us consider an open quantum system described by a Hamiltonian $\textbf{H}$ on a Hilbert space $\hilb=\hilb_{\rm S}\otimes\hilb_{\rm B}$, with $\hilb_{\rm S}$ being the space associated with the experimentally accessible system and $\hilb_{\rm B}$ with the external environment (or bath). Suppose that said system is initially prepared at a time $t_0$ in a state associated with some density operator $\varrho_0\in\mathcal{B}(\hilb)$, and repeatedly probed at times $t_n\geq t_{n-1}\geq\dots t_1\geq t_0$ via some family of sharp projective measurements in an orthonormal basis $\{\ket{x}\}_{x\in\mathcal{X}}$ of $\hilb_{\rm S}$. The laws of quantum mechanics provide us a simple and elegant rule to construct the joint probability distribution associated with any such sequence of measurements:
	\begin{equation}\label{eq:probquantum}
		\mathbb{P}_n(x_n,t_n;\dots;x_1,t_1)=\tr\left[\left(\mathcal{P}_{x_n}\otimes\mathcal{I}_{\rm B}\right)\mathcal{U}_{t_n-t_{n-1}}\cdots\left(\mathcal{P}_{x_1}\otimes\mathcal{I}_{\rm B}\right)\mathcal{U}_{t_1-t_0}(\varrho_{0})\right],
	\end{equation}
	with $\mathcal{U}_t=\e^{-\i t\textbf{H}}(\cdot)\e^{\i t\textbf{H}}$ representing the free evolution of the system+environment induced by $\textbf{H}$ on $\hilb_{\rm S}\otimes\hilb_{\rm B}$,  $\mathcal{I}_{\rm B}$ the identity map on the environment space $\hilb_{\rm B}$,
	and $\mathcal{P}_x = P_x (\cdot)P_x$, with $P_x = |x\>\!\<x|$. Eq.~\eqref{eq:probquantum} defines a legitimate family of joint probability distributions; however, in general they \textit{do not satisfy} the consistency condition~\eqref{eq:consistency}, and thus are not associated with any underlying classical stochastic process. This is ultimately related to the fact that, in general, measurement \textit{do} disturb the state of a quantum system---not performing a measurement is not the same as performing said measurement and forgetting it, in contrast with Eq.~\eqref{eq:consistency}.
	
	Inspired by the theory of stochastic processes, and following the same point of view as in~\cite{Andrea-1,Andrea-2}, we shall adopt hereafter the following definitions. Given an open quantum system associated with a total Hamiltonian $\textbf{H}$ on $\hilb=\hilb_{\rm S}\otimes\hilb_{\rm B}$, initially prepared in some collective state $\varrho_0$, and a fixed orthonormal basis of $\hilb_{\rm B}$, we shall denote as \textit{quantum process} any sequence of sharp measurements in said basis on the system; defining the corresponding family of joint probability distributions as in Eq.~\eqref{eq:probquantum}, the process is said to be
	\begin{itemize}
		\item \textit{Markovian}, if the Markov property~\eqref{eq:markov} holds.
		\item \textit{classical}, if the consistency condition~\eqref{eq:consistency} (or, in the Markov case, Eq.~\eqref{eq:chapman}) holds.
	\end{itemize}
	Importantly, both definitions are strictly \textit{dependent} on the choice of the measurement basis.
	
	In this regard, it is worth recalling that no universally accepted definition of quantum Markovianity exists. Diverse mathematical properties, each focusing on particular aspects of the problem under investigation, have been---more or less formally---put forward as definitions of quantum Markovianity. Some of them solely involve the properties of the reduced dynamics: for instance, the completely positive divisibility~\cite{RHP} (and, as a particular case, the semigroup property) and the monotonicity of the distinguishability between arbitrary states~\cite{BLP}. Other approaches take into account the full unitary dynamics of the system \textit{and} the environment in order to account for the unraveling of information backflow via external interventions or measurements on the system: among many other, the validity of the factorization (or Born) approximation---roughly speaking, the idea that the environment, not ``feeling'' the action of the system, evolves independently of the presence of the latter---or the quantum regression formula. An exhaustive discussion of all concepts of Markovianity in quantum mechanics and the hierarchical relations between them is reported in~\cite{NM4}. One should always take into account this ambiguity when dealing with concepts of (non--)Markovianity in open quantum systems (cf.\ the recent reviews~\cite{NM1,NM2,NM3,PR}). The approach of~\cite{Andrea-1,Andrea-2} adopted in this paper is closely related to the mathematical formulation of quantum Markov stochastic processes proposed in~\cite{QP1,QP2,QP3} (cf.\ also the recent review~\cite{Modi-PRX}). Moreover, it is closely related to the recent approach to quantum Markovianity proposed in~\cite{kavan1,kavan2,kavan3} in which the Markovianity of the corresponding process is characterized in terms of the so-called quantum process tensor of the system. The factorization of the process tensor is essentially equivalent to the validity of quantum regression~\cite{NM4}. For a discussion of Markovianity based on the quantum regression formula see also~\cite{Francesco,Bassano}.
	
	Adopting the above definition of Markovianity~\cite{Andrea-1,Andrea-2}, a useful characterization of classicality was first provided in~\cite{Andrea-1} for dynamical semigroups, and then extended in~\cite{Andrea-2} to general quantum Markov processes. Classicality was shown to be crucially interrelated with the generation of coherence: fundamentally, a Markovian multitime statistics fails to be classical if and only if the dynamics generates coherences and subsequently turns them into populations---mathematically, classicality holds whenever the process can be represented by means of non-coherence-generating-and-detecting (NCGD) maps. Furthermore, in~\cite{Andrea-2} a similar characterization was obtained for (possibly) non-Markovian processes: while in this case the absence of coherence does not guarantee classicality, they managed to provide a direct connection between classicality and the vanishing of quantum discord between the evolving system and its environment.
	
	Inspired by such general results, the scope of the present paper is to provide an explicit discussion of Markovianity and classicality in the context of a specific, paradigmatic class of open quantum systems: spin boson models and their generalizations. The interest in such models is not recent. Since decades, spin--boson models and their generalizations have been claiming a primary role in the theory of open quantum systems~\cite{Breuer,Legget,Weiss,RIVAS,Ingold}, finding applications as toy models for the description of several quantum phenomena like noise and decoherence~\cite{Zoller,Plenio,Car,Zurek,DEC1,DEC2} as well as practical applications~\cite{Zoller,Plenio,Car,Breuer,RIVAS,SIM,phase}. The Markovianity of the evolution induced by spin--boson models, in many of the possible declensions listed above, has also been analyzed by several authors~\cite{Breuer,RIVAS,NM1,NM2,NM3}; spin--boson models also provide simple examples in which, via a suitable choice of the form factor, the onset of non-Markovianity can be delayed at arbitrarily large times~\cite{hidden1,hidden2}.
	
	In this regard, we shall consider a $(d+1)$-level quantum system living in the Hilbert space $\hilb_{\rm S}=\hilb_{\mathrm{e}}\oplus\hilb_{\mathrm{g}}$, with $\dim\hilb_{\mathrm{e}}=d$ and $\dim\hilb_{\mathrm{g}}=1$. $\hilb_{\mathrm{e}}$ corresponds to a $d$-dimensional excited sector, whereas $\hilb_{\mathrm{g}}$ is spanned by the ground state $|\mathrm{0}\>$. The system is coupled to a multimode boson bath, which (for simplicity) will be taken as a $d$-mode bath, with the total system-bath Hamiltonian given by ($\hbar=1$)
	\begin{equation}\label{Hn}
		\mathbf{H} = H_{\mathrm{e}} \otimes \oper_{\rm B} + \oper_{\rm S} \otimes \sum_{j=1}^d \int \mathrm{d}\omega\;\omega\, b^\dagger_j(\omega) b_j(\omega) + H_{\rm int} ,
	\end{equation}
	where $H_{\e}$ is the free Hamiltonian of the excited sector of the system, and the interaction term reads	
	\begin{equation}\label{hnint}
		H_{\rm int} = \sum_{j=1}^d\int \mathrm{d}\omega\;f_j(\omega) |0\>\!\<\e_j| \otimes b^\dagger_{j}(\omega) + \rm{h.c.},
	\end{equation}
	with $|\e_1\>,\dots,|\e_d\> \in \hilb_{\e}$ being a collection of linearly independent vectors in $\mathcal{H}_\e$, and where h.c. stands for the Hermitian conjugated term; each function $f_j(\omega)$ (form factor) modulates the coupling between the system and the $j$th mode of the bath. The boson creation and annihilation operators $b^\dagger_j(\omega)$ and $b_j(\omega)$ satisfy the standard canonical commutation relations: $[b_i(\omega),b_{j}(\omega')]=0$ and $[b_i(\omega),b^\dag_{j}(\omega')]=\delta_{ij}\delta(\omega-\omega')$, to be interpreted in the distributional sense. This model was first introduced in Ref.~\cite{garraway}, and further studied in~\cite{Davide-1}; it belongs to the class of generalized spin--boson models~\cite{arai,arai2,arai3,hirokawa,hirokawa2}, its self-adjointness being ensured whenever all functions $f_j$ are normalizable (square-integrable)~\cite{arai} but also for suitable classes of non-normalizable form factors~\cite{gsb}.
	
	For such models, after recalling the main properties of the unitary evolution in the single-excitation sector, we shall provide an explicit, constructive discussion of Markovianity and classicality in the limit in which all form factors are flat---that is, all boson frequencies, positive and negative, are coupled to the atom with the same strength. Precisely, in such a limit the model will be proven to be Markovian---in the sense of Eq.~\eqref{eq:markov}---with respect to measurements in an \textit{arbitrary} orthogonal basis. Furthermore, necessary and sufficient conditions for its classicality in this limit will be found, and discussed at the light of the existing general results in~\cite{Andrea-1,Andrea-2} about classicality in quantum processes.
	
	The paper is structured as follows. In Section~\ref{sec:general} we revise the exact computation of the dynamics of any state in the form $\ket{\Psi_0}=\ket{\psi_0}\otimes\ket{\vac}\in\hilb_{\rm S}\otimes\hilb_{\rm B}$ generated by $\textbf{H}$, derive the corresponding reduced dynamics on $\hilb_{\rm S}$, and study its divisibility properties. In Section~\ref{sec:flat} we take into account the limiting case of flat atom-field couplings, studying the corresponding properties of the reduced dynamics. Finally, Section~\ref{sec:marklass} is devoted to the main results of the paper: proving that the model is indeed Markovian, in the sense discussed above, with respect to arbitrary sharp measurements, and proving simple conditions under which it is (non--)classical. Some concluding remarks are outlined in Section~\ref{sec:conclusions}.
	
	\section{Generalities on the multilevel spin--boson model}\label{sec:general}
	
	In the following, without loss of generality we will set the initial time of the evolution induced by $\textbf{H}$ as $t_0=0$ unless otherwise stated.
	
	\subsection{Unitary and reduced dynamics}
	
	As discussed in~\cite{Davide-1} (see also~\cite{FL,hidden1} for the spin--boson case), the dynamics of an arbitrary state in the form
	\begin{equation}\label{ini}
		\ket{\Psi_0} = \ket{\psi_0}\otimes\ket{\vac},\;\;\text{with}\quad \ket{\psi_0}=\alpha |0\> + |\psi_{\e}\>,\;\;\alpha\in\mathbb{C},\;\ket{\psi_\e}\in\hilb_{\e}
	\end{equation}
	induced by the Hamiltonian $\textbf{H}$ in Eqs.~\eqref{Hn}--\eqref{hnint}, can be computed exactly. One finds that the time evolved state $\ket{\Psi_t}:=\e^{-\i t\mathbf{H}}\ket{\Psi_0}$ has the following form:	
	\begin{equation}
		\ket{\Psi_t}= \Big[\alpha |0\> + \ket{\psi_\e(t)}\Big] \otimes \ket{\vac} + \ket{0} \otimes \sum_{j=1}^d \int \mathrm{d}\omega\: \xi_j(t,\omega)  b^\dagger_j(\omega) \ket{\vac},
	\end{equation}
	with $|\psi_\e(t)\>\in\hilb_{\e}$ and the wavefunctions $\xi_j(t,\omega)$ satisfying the following system of equations, equivalent to the Schr\"odinger equation generated by $\textbf{H}$:
	\begin{eqnarray}\label{S1}
		\i |\dot{\psi}_\e(t)\> &=& H_\e |\psi_\e(t)\> +  \sum_{j=1}^d \int \mathrm{d}\omega\, f^*_j(\omega) \xi_j(t,\omega) |\e_j\> ,\\
		\i\, \dot{\xi}_j(t,\omega) &=& \omega\,\xi_j(t,\omega) + f_j(\omega) \< \e_j|\psi_\e(t)\>  , \ \ \ j=1,\ldots,d .\label{S2}
	\end{eqnarray}
	Solving Eq.~\eqref{S2}: 	
	\begin{equation}
		\xi_j(t,\omega) = -\i \int_0^t \mathrm{d}s\:\e^{-\i\omega(t-s)} f_j(\omega) \<\e_j|\psi_\e(s)\> ,
	\end{equation}
	and inserting the solution into Eq.~\eqref{S1}, one obtains	
	\begin{equation}
		|\psi_\e(t)\> = \A(t) |\psi_\e\> ,
	\end{equation}
	where the time-dependent operator $\A(t)\in\mathcal{B}(\hilb_{\e})$ satisfies the following non-local equation:	
	\begin{equation}\label{S1a}
		\i \dot{\A}(t) = H_\e \A(t) +   \int_0^t \mathrm{d}s\;\mathrm{G}(t-s)\A(s)  , \ \ \ \A(0) = \oper_\e ,
	\end{equation}
	and  the corresponding memory kernel $\mathrm{G}(t)\in\mathcal{B}(\hilb_{\e})$ has the following form:	
	\begin{equation}\label{eq:generator}
		\mathrm{G}(t) = -\i \sum_{j=1}^d \int \mathrm{d}\omega\:\e^{-\i\omega t} |f_j(\omega)|^2 |\e_j\>\!\<\e_j| .
	\end{equation}
	By using these properties, one easily finds the corresponding reduced evolution of the state of the system:	
	\begin{eqnarray}
		|\psi_0\>\!\<\psi_0| \to \Lambda_t(|\psi_0\>\!\<\psi_0|) =  {\rm Tr}_B \, |\Psi_t\>\!\<\Psi_t| &=& \A(t) |\psi_\e\>\!\<\psi_e|\A^\dagger(t) + \A(t) |\psi_\e\>\!\<0| + |0\>\!\<\psi_\e|\A^\dagger(t) \nonumber \\ &+& |0\>\!\<0| \Big( |\alpha|^2 + \<\psi_\e|\oper_e - \A^\dagger(t)\A(t)|\psi_e\> \Big).
	\end{eqnarray}
	By construction, the map $\Lambda_t$ is completely positive and trace preserving (CPTP) for all $t \geq 0$, and satisfies $\Lambda_{t=0} = \mathcal{I}_S$. Using a natural splitting of the system density operator,
	\begin{equation}
		\rho= \left( \begin{array}{cc} \hat{\rho}_\e & |\mathbf{w}\> \\  \< \mathbf{w}| & \rho_\mathrm{g} \end{array} \right) ,
	\end{equation}
	one finds the following expression for the reduced evolution:
	\begin{equation}\label{MAD}
		\Lambda_t(\rho) = \left( \begin{array}{cc} \A(t)\hat{\rho}_\e \A^\dagger(t) & \A(t)|\mathbf{w}\> \\  \< \mathbf{w}|\A^\dagger(t) & \rho_\mathrm{g}(t) \end{array} \right) ,
	\end{equation}
	with $\rho_\mathrm{g}(t) = {\rm Tr}\,\rho -  {\rm Tr}( \A(t) \hat{\rho}_\e \A^\dagger(t))$. It is therefore clear that the entire dynamical map $\Lambda_t$ is uniquely characterized by the operator $\A(t)$, which we shall denote as the \textit{survival amplitude operator}, since the quantity $\| \A(t)|\psi_\e\>\|^2$ defines the probability that the initial state $|\psi_\e\> \otimes |\vac\>$ does not decay to the ground state at the time $t$. By construction, the following normalization condition holds:	
	\begin{equation}
		\| \A(t)|\psi_\e\>\|^2 + \sum_{j=1}^d \|\xi_j(t)\|^2 =  \| |\psi_\e\>\|^2 ,
	\end{equation}
	which accounts for the property ${\rm Tr}\,\Lambda_t(\rho) = {\rm Tr}\,\rho$ ensuring the preservation of the trace.
	
	The operator $\A(t) : \mathcal{H}_\e \to \mathcal{H}_\e$ reduces to a single function $a(t)$ in the qubit case and Eq.~\eqref{MAD} reduces to the well-known amplitude-damping qubit channel,	
	\begin{equation}
		\Lambda_t(\rho) = \left( \begin{array}{cc} |a(t)|^2 {\rho}_{11} & a(t) \rho_{10} \\  a^*(t) \rho_{01} & \rho_{00} + (1-|a(t)|^2)\rho_{11}  \end{array} \right);
	\end{equation}
	as such, we shall refer to $\Lambda_t$ as a \textit{multilevel amplitude-damping channel}. We point out that other generalizations of the amplitude-damping channel were analyzed in Refs.~\cite{giovannetti,wilde}; besides, a family of quantum channels further generalizing the structure~\eqref{MAD}, labeled as \textit{excitation-damping channels}, was recently introduced and studied in~\cite{Davide-3}.
	
	The following simple characterization of positivity and complete positivity for channels in the form~\eqref{MAD} was proven in~\cite{Davide-1}:	
	\begin{proposition}[\!\!\cite{Davide-1}] The dynamical map~\eqref{MAD} is completely positive if and only if $\A(t)$ is a contraction w.r.t. the operator norm, i.e. $\|\A(t)\|_{\rm op} \leq 1$ for all $t \geq 0$. Besides, it is completely positive if and only if it is positive.
	\end{proposition}
	As such, channels as in Eq.~\eqref{MAD} represent an interesting example of channels for which complete positivity and positivity coincide, another important instance being multilevel dephasing channels (or Hadamard channels)~\cite{Benatti,Davide-2}.
	
	\subsection{Markovianity and divisibility of the model}
	
	Recall that a dynamical map $\Lambda_t$ is divisible whenever for any $t \geq s$ one has $\Lambda_t = \Lambda_{t,s} \Lambda_s$ for some corresponding propagator $\Lambda_{t,s} : \mathcal{B}(\mathcal{H}) \to \mathcal{B}(\mathcal{H})$; it is CP-divisible when $\Lambda_{t,s}$ is CPTP and P-divisible when $\Lambda_{t,s}$ is positive and trace-preserving. CP-divisibility is a common core of the various (and generally inequivalent) definitions of \textit{Markovianity} in open quantum systems. Clearly, every invertible map is also divisible with
	\begin{equation}\label{eq:lambdats}
		\Lambda_{t,s} = \Lambda_t \Lambda_s^{-1},
	\end{equation}
	and is thus CP-divisible (resp. P-divisible) if and only if the map in Eq.~\eqref{eq:lambdats} is completely positive (resp. positive) for all $t\geq s$. In particular, invertible and differentiable channels satisfy the master equation $\dot\Lambda_t(\rho)=\mathcal{L}_t\left(\Lambda_t(\rho)\right)$, with the time-local generator $\mathcal{L}_t=\dot\Lambda_t\Lambda_t^{-1}$, and it is known that $\Lambda_t$ is CP-divisible if and only if $\mathcal{L}_t$ is a time-dependent Gorini--Kossakowski--Lindblad--Sudarshan (GKLS) map \cite{GKS,L}.
	
	In our case, a simple computation shows that the multilevel amplitude-damping channel $\Lambda_t$ as defined in Eq.~\eqref{MAD} is invertible if and only if the corresponding survival amplitude operator $\A(t)$ is invertible. In such a case, the channel is divisible with
	\begin{equation}
		\Lambda_{t,s}(\rho) = \left( \begin{array}{cc} \A(t,s)\hat{\rho}_\e \A^\dagger(t,s) & \A(t,s)|\mathbf{w}\> \\  \< \mathbf{w}| \A^\dagger(t,s) & \rho_\mathrm{g}(t,s) \end{array} \right)
	\end{equation}
	where $\A(t,s) = \A(t)\A^{-1}(s)$ and  $\rho_\mathrm{g}(t,s) = {\rm Tr}\rho - {\rm Tr}[\A(t,s)\hat{\rho}_\e \A^\dagger(t,s)] $. This readily implies the following characterization of CP-divisibility and P-divisibility.	
	\begin{proposition}[\!\!\cite{Davide-1}] Let the amplitude operator $\A(t)$ be invertible for all $t \geq 0$. Then $\Lambda_t$ is CP-divisible if and only if		
		\begin{equation}\label{eq:contraction}
			\| \A(t,s) \|_{\rm op} \leq 1 ,
		\end{equation}
		for all $t \geq s$, and it is P-divisible if and only if it is CP-divisible.
	\end{proposition}
	Interestingly, CP-divisibility and P-divisibility are equivalent properties for the multilevel amplitude-damping channel. We also notice that the property~\eqref{eq:contraction} is equivalent, provided that $t\mapsto\A(t)$ is differentiable, to the following one:
	\begin{equation}\label{eq:contraction2}
		\frac{\mathrm{d}}{\mathrm{d}t}\, \| \A(t)|\psi_\e\> \| \leq 0 ,
	\end{equation}
	for any $|\psi_\e\> \in \mathcal{H}_\e$. In other words:	
	\begin{corollary}
		Let $\Lambda_t$ being invertible and differentiable. Then $\Lambda_t$ is a monotonic contraction w.r.t. trace norm in $\mathcal{B}(\mathcal{H})$ if and only if $\A(t)$ is a monotonic contraction in $\mathcal{H}_\e$ (w.r.t. the natural norm induced by the inner product in $\mathcal{H}_\e$).		
	\end{corollary}
	The time-dependent generator associated with the channel reads
	\begin{equation}\label{L}
		\mathcal{L}_t(\rho) = \left( \begin{array}{cc} \L(t)\hat{\rho}_\e + \hat{\rho}_\e \L^\dagger(t) & \L(t)|\mathbf{w}\> \\  \< \mathbf{w}|\L^\dagger(t) & -{\rm Tr}(\hat{\rho}_\e[ \L(t) + \L^\dagger(t)]) \end{array} \right) ,
	\end{equation}
	where	
	\begin{equation}
		\L(t) = \dot{\A}(t)\A^{-1}(t) ,
	\end{equation}
	i.e. $\A(t)$ satisfies the following time-local dynamical equation	
	\begin{equation}
		\dot{\A}(t)  = \L(t)\A(t) \ , \ \ \ \A(0) =\oper_\e ,
	\end{equation}
	and hence the corresponding propagator in $\mathcal{H}_\e$ formally reads	
	\begin{equation}
		\A(t,s) = \mathcal{T}\exp\Bigl( \int_{s}^t \L(\tau)\,\mathrm{d}\tau \Bigr),
	\end{equation}
	with $\mathcal{T}$ signaling time ordering. Finally, by defining the time-dependent Hermitian operators
	\begin{equation}
		\H_\e(t) := \frac{\i}{2} [\L(t) - \L^\dagger(t)],\qquad 	{\Gamma}(t) := - [\L(t) + \L^\dagger(t)] ,
	\end{equation}
	the corresponding time-local generator can be written as
	\begin{equation}
		\mathcal{L}_t(\rho) = -\i[\H_\e(t),\rho] - \frac 12 \{ \Gamma(t), \rho\} + {\rm Tr}(\Gamma(t) \rho) \, |0\>\!\<0|  ,
	\end{equation}
	which is manifestly a GKLS generator provided that $\Gamma(t)\geq0$, that is, the operator $\Gamma(t)$ is positive semidefinite at all times. Consequently,
	\begin{proposition}
		The reduced evolution induced by a multilevel spin--boson model is CP-divisible if and only if $\Gamma(t)\geq0$ for all $t\geq0$.
	\end{proposition}
	
	\section{The flat coupling limit}\label{sec:flat}
	
	We shall focus hereafter on a particular choice of the form factors $f_j(\omega)$ modulating the coupling between the system and the $j$th mode of the field. Namely, we shall consider the case of \textit{flat} couplings on all (positive and negative) values of $\omega$:
	\begin{equation}\label{eq:flat}
		|f_j(\omega)|^2 = \frac{\gamma_j}{2\pi}  , \quad-\infty\leq\omega\leq\infty,
	\end{equation}
	with $\gamma_j\geq0$ being a coupling constant. This choice of coupling corresponds, in the position representation, to the case in which the coupling between the system and each mode of the field is described by a Dirac delta interaction; as such, it provides an idealized description of short-range system-field interactions. Clearly, such form factors are not normalizable, whence the self-adjointness of the corresponding multilevel spin--boson model is not guaranteed \textit{a priori}; these mathematically issues are thoroughly discussed in greater generality in~\cite{FL,gsb}. For our purposes, here it will suffice to observe that the equations of motion~\eqref{S1}--\eqref{S2} induced by the model are indeed well-posed in this limit. Indeed, substituting Eq.~\eqref{eq:flat} into the expression~\eqref{eq:generator}, the memory kernel $\mathrm{G}(t)$ collapses into a Dirac delta:
	\begin{equation}
		\mathrm{G}(t-s)=-\i\delta(t-s)\sum_{j=1}^d\gamma_j\ketbra{\e_j}{\e_j},
	\end{equation}
	that is, Eq.~\eqref{S1a}, and thus both Eqs.~\eqref{S1}--\eqref{S2}, are \textit{memoryless}: they do not involve any dependence on past times.
	
	\subsection{Dynamics and reduced dynamics}
	Specifically, Eq.~\eqref{S1a} for the survival amplitude operator simplifies as follows:
	\begin{equation}
		\i\dot\A(t)=\left(\H_e-\frac{\i}{2}\Gamma\right)\A(t),
	\end{equation}
	where
	\begin{equation}
		\Gamma = 2\sum_{j=1}^d \gamma_j \ketbra{\e_j}{\e_j}\geq0
	\end{equation}
	and its solution simply reads
	\begin{equation}\label{eq:exp}
		\A(t)=\e^{-t\left(\i\H_\e+\frac{1}{2}\Gamma\right)}.
	\end{equation}
	Consequently, the solution of the system~\eqref{S1}--\eqref{S2} reads as follows:
	\begin{eqnarray}\label{eq:evol_psi}
		\ket{\psi_\e(t)}&=&\e^{-t\left(\i\H_\e+\frac{1}{2}\Gamma\right)}\ket{\psi_\e},\\
		\label{xi-A}
		{\xi}_j(t,\omega) &=& -\i \sqrt{\frac{\gamma_j}{2\pi}} \int_0^t \mathrm{d}s \;\e^{-\i\omega(t-s)} \Braket{\e_j|\e^{-s\left(\i \H_\e+\frac{1}{2}\Gamma\right)}|\psi_e}.
	\end{eqnarray}
	Two important remarks follow.
	
	\paragraph{Wavefunction in the position representation.} Interestingly, Eq.~\eqref{xi-A} corresponds, in the position representation, to a compactly supported wavefunction. Indeed, defining the  Fourier antitransform of $\xi_j(t,\omega)$ via $
	\hat{\xi}_j(t,x)=\int\mathrm{d}\omega\;\e^{\i\omega x}\xi_j(t,\omega)
	$, one easily finds
	\begin{equation}\label{eq:propa}
		\hat{\xi}_j(t,x) = -\i \sqrt{2\pi \gamma_j} \int_0^t \mathrm{d}s \Braket{\e_j|\e^{-(t-x)\left(\i \H_\e+\frac{1}{2}\Gamma\right)}|\psi_e} \, \chi_{[0,t]}(x) ,
	\end{equation}
	where $\chi_I(x)$ is the characteristic function of the interval $I$. Hence, $\hat{\xi}_j(t,x)$ is {entirely supported} on $[0,t]$, for positive times $t>0$. This result is analogous to the one obtained (in the single-atom case) in~\cite{hidden2}  and has the following intuitive explanation: preparing the system in the state $\ket{\psi}\otimes\ket{\vac}$ (multilevel atom in a generic state of the excited sector + all modes of the boson field in the vacuum state), the system will be in a superposition of a state in which no bosons are emitted, plus all possible states in which the system decays and a boson is emitted at the position $x=0$ in \textit{only one} of all possible modes: for each mode, the emitted boson moves from left to right at unit speed (since we keep $c=1$), whence the wavefunction at time $t$ will be supported at $[0,t]$. This property of $\xi_j(t,\omega)$ will be crucial to prove Prop.~\ref{PRO-A}, which, in turn, will play important role in the analysis of Markovianity of the quantum measurement process in the next section.
	
	\paragraph{Case $[\H_\e,\Gamma]=0$.} In this case the vectors $\{\ket{\e_j}\}_j$ are \textit{eigenvectors} of $\Gamma$ and thus $\H_\e=\sum_j\omega_j\ketbra{\e_j}{e_j}$, whence Eqs.~\eqref{eq:evol_psi}--\eqref{xi-A} read explicitly
	\begin{eqnarray}
		\ket{\psi_\e(t)}&=&\sum_\ell\braket{\e_\ell|\psi_\e}\e^{-t\left(\i \omega_\ell+\frac{1}{2}\gamma_\ell\right)}\ket{\e_\ell},\\
		{\xi}_j(t,\omega) &=& -\i \sqrt{\frac{\gamma_j}{2\pi}} \braket{\e_j|\psi_\e}\int_0^t \mathrm{d}s \;\e^{-\i\omega(t-s)} \e^{-s\left(\i \omega_j+\frac{1}{2}\gamma_j\right)},
	\end{eqnarray}
	that is, the survival operator is simply given by
	\begin{equation}\label{eq:at-diag}
		\A(t)=\sum_\ell\e^{-t\left(\i \omega_\ell+\frac{1}{2}\gamma_\ell\right)}\ketbra{\e_\ell}{\e_\ell};
	\end{equation}
	furthermore, in the position representation,
	\begin{equation}
		\hat{\xi}_j(t,x) = -\i \sqrt{2\pi \gamma_j}\braket{\e_j|\psi_\e} \int_0^t \mathrm{d}s\;\e^{-(t-x)\left(\i \omega_j+\frac{1}{2}\gamma_j\right)}\, \chi_{[0,t]}(x).
	\end{equation}
	In particular, if $\ket{\psi_\e}=\ket{\e_\alpha}$ for some fixed $\alpha=1,\dots,d$, then
	\begin{eqnarray}
		\ket{\psi_\e(t)}&=&\e^{-t\left(\i \omega_\alpha+\frac{1}{2}\gamma_\alpha\right)}\ket{\e_\alpha},\\
		{\xi}_j(t,\omega) &=& -\i \sqrt{\frac{\gamma_j}{2\pi}} \delta_{j\alpha}\int_0^t \mathrm{d}s\;\e^{-(t-x)\left(\i \omega_j+\frac{1}{2}\gamma_j\right)}\, \chi_{[0,t]}(x),
	\end{eqnarray}
	that is, the evolution of the system only involves $\ket{\e_\alpha}$ and the $\alpha$th mode of the boson field, effectively behaving as a spin--boson model. The multilevel, multimode nature of the model only emerges when taking into account different initial states.
	
	As a final observation, notice that the equality $[\H_\e,\Gamma]=0$ holds if and only if $\A(t)$ is a \textit{normal} matrix, i.e. $[\A(t),\A^\dag(t)]=0$. Indeed, by the properties of matrix exponentials,
	\begin{equation}
		[\A(t),\A^\dag(t)]=0\;\;\forall t\in\mathbb{R}\iff0=\left[\i\H_\e+\frac{1}{2}\Gamma,-\i\H_\e+\frac{1}{2}\Gamma\right]=\i[\H_\e,\Gamma].
	\end{equation}
	
	Finally, as for the reduced dynamics, the corresponding channel $\Lambda_t$ is a \textit{semigroup}, i.e. $\Lambda_{t}=\Lambda_{t-s}\Lambda_s$, as can be either seen by directly substituting Eq.~\eqref{eq:exp} into the expression~\eqref{MAD} of the channel or, equivalently, by simply observing that the expression~\eqref{L} of its time-local generator reduces to
	\begin{eqnarray}
		\mathcal{L}_t(\rho)\equiv\mathcal{L}(\rho)&=&-\i[\H_\e,\rho] - \frac 12 \{ \Gamma, \rho\}+ {\rm Tr}(\Gamma \rho) \, |0\>\!\<0|  \nonumber\\
		&=&  -\i[\H_\e,\rho] + \sum_{j=1}^d \gamma_j \Big( |0\>\!\<\e_j|\rho|\e_j\>\!\<0| - \frac 12 \{ |\e_j\>\!\<\e_j|,\rho\} \Big),
	\end{eqnarray}
	which is clearly a time-independent GKLS generator.	\textit{A fortiori}, the channel is CP-divisible at all times.	
	
	\subsection{Other properties}
	
	In general, as discussed, the dynamics induced by $\textbf{H}$ on states with a number of excitations \textit{larger} than one is not solvable. For instance, the evolution of states with two excitations, like
	\begin{equation}
		\e^{-\i t\textbf{H}}\ket{\psi_\e}\otimes b_j^\dag(\eta)\ket{\vac}\quad\text{or}\quad\e^{-\i t\textbf{H}}\ket{0}\otimes b_j^\dag(\eta)b_{j'}^\dag(\eta')\ket{\vac},
	\end{equation}
	cannot be computed explicitly. This, in particular, would generally prevent us from computing the joint probability distributions associated with the quantum process obtained by repeatedly probing the system via any orthonormal basis: a measurement in a basis which is not compatible with the $\hilb_{\e}\oplus\hilb_\g$ splitting, as we will see, will necessarily involve taking into account higher-excitation sectors.
	
	However, \textit{specifically} in the flat coupling case, a fundamental property holds:
	\begin{proposition}   \label{PRO-A} For $\xi_j(t,\omega)$ defined in Eq.~\eqref{xi-A}, one has
		\begin{equation}
			U_{\tau} \oper_S \otimes b^\dagger_j(\xi_j(t)) U^\dagger_{\tau} = \oper_S \otimes b^\dagger_j(\xi_j(\tau,t)) ,
		\end{equation}
		for any $\tau \geq 0$, where we define
		\begin{equation}
			\xi_j(\tau,t,\omega) :=\e^{-\i\omega \tau}\xi_j(t,\omega).
		\end{equation}
	\end{proposition}
	The proof of Prop.~\ref{PRO-A} is reported in Appendix~\ref{app:1}. This property will be used in the following section to compute the multitime statistics associated with the process both with respect to measurements in bases compatible with the $\hilb_{\e}\oplus\hilb_\g$ splitting (cf.\ Section~\ref{subsec:compatible}) and arbitrary bases (cf.\ Section~\ref{subsec:arbitrary}).
	
	\section{Markovianity and classicality}\label{sec:marklass}
	
	The limiting choice of flat couplings in the interaction Hamiltonian~\eqref{hnint} is the only possible choice ensuring that the reduced dynamics, associated with the map $\Lambda_t$, is a GKLS semigroup---historically, such semigroups have long been denoted as \textit{Markovian semigroups}. However, the semigroup property by itself does \textit{not} suffices, in general, to ensure Markovianity of the corresponding process, in the sense of Eq.~\eqref{eq:markov}, with respect to an arbitrary sharp measurement.
	
	In this regard, we shall consider projective measurements w.r.t. an orthonormal basis $\{|x_\alpha\>\}_{x=0,\dots,d}$ in the  Hilbert space $\mathcal{H}_{\rm S}$. We shall assume the initial state of the system+bath to be $|\psi_0\> \otimes |\vac\>$. With this choice, the multitime statistic generated by this process, cf.\ Eq.~\eqref{eq:probquantum}, is
	\begin{equation}\label{I}
		\mathbb{P}_n(x_n,t_n;\ldots ;x_1,t_1) = {\rm Tr} [ \mathcal{P}_{x_n}\otimes \mathcal{I}_B \circ U_{t_n-t_{n-1}} \ldots \circ \mathcal{U}_{t_2-t_1} \circ \mathcal{P}_{x_1}\otimes \mathcal{I}_B \circ U_{t_1}(|\psi_0\>\!\<\psi_0| \otimes |\vac\>\!\<\vac|)] ,
	\end{equation}
	where $\mathcal{P}_x(X) = P_x X P_x$, $P_x=\ketbra{x}{x}$.
	
	In this section we will provide an explicit proof of the fact that, indeed, this quantum process is Markovian in the limit of flat coupling. Furthermore, necessary and sufficient conditions for the classicality of the model---that is, for the Chapman--Kolmogorov equation~\eqref{eq:chapman} to be satisfied---will be found. For the sake of simplicity, we shall start by examining the situation in which the measurements are performed with respect to an orthonormal basis compatible with the splitting $\hilb_{\rm S}=\hilb_\e\oplus\hilb_\g$, and then examine the general situation.
	
	\subsection{Interlude: Markovianity and quantum regression}
	
	Before going on, it will be useful to briefly discuss the link between Markovianity and the validity of the general quantum regression formula (GQRF)~\cite{Lax,Breuer,NM4}. Given any open quantum system and a fixed state of the bath $\rho_{\rm B}$, the unitary evolution of the ``system + bath'' $\mathcal{U}_t \rho = U_t \rho U_t^\dagger$ giving rise to a CP-divisible reduced evolution of the system $\Lambda_t = {\rm Tr}_{\rm B}\,\mathcal{U}_t(\rho \otimes\rho_{\rm B})$, we say that a pair $(\mathcal{U}_t,\rho_{\rm B})$ satisfies the GQRF if, for any collection of times $t_n > t_{n-1} > \ldots > t_1 > 0$ and two sets of system operators $\{X_0,X_1,\ldots,X_n\}$ and  $\{Y_0,Y_1,\ldots,Y_n\}$, one has the following relation between multitime correlation functions:
	\begin{equation}\label{QR0}
		{\rm Tr}[\tilde{\mathcal{E}}_n \mathcal{U}_{t_n-t_{n-1}} \cdots \tilde{\mathcal{E}}_1\,\mathcal{U}_{t_1}(\rho \otimes \rho_{\rm B})] = {\rm Tr}[\mathcal{E}_n \Lambda_{t_n,t_{n-1}} \cdots \mathcal{E}_1 \Lambda_{t_1}(\rho)] ,
	\end{equation}
	where $\mathcal{E}_k=X_k\cdot Y_k$ and $\tilde{\mathcal{E}}_k=\mathcal{E}_k\otimes\mathcal{I}_{\rm B}$. Eq.~\eqref{QR0} means that all correlation functions for the ``system+bath'' evolution can be computed in terms of the dynamical map of the system alone. In particular, in such a case,
	\begin{eqnarray}\label{eq:probquantum_regression}
		\mathbb{P}_n(x_n,t_n;\dots;x_1,t_1)&=&\tr\left[\left(\mathcal{P}_{x_n}\otimes\mathcal{I}_{\rm B}\right)\mathcal{U}_{t_n-t_{n-1}}\cdots\left(\mathcal{P}_{x_1}\otimes\mathcal{I}_{\rm B}\right)\mathcal{U}_{t_1}(\rho\otimes\rho_{\rm B})\right]\nonumber\\
		&=&\tr\left[\mathcal{P}_{x_n}\Lambda_{t_n,t_{n-1}}\cdots\mathcal{P}_{x_1}\Lambda_{t_1}(\rho)\right]\nonumber\\
		&=&\bra{x_1}\Lambda_{t_1}(\rho)\ket{x_1}
		\prod_{j=1}^{n-1}
		\bra{x_{j+1}}
		\Lambda_{t_{j+1},t_{j}}
		\left(
		\ketbra{x_{j}}{x_{j}}
		\right)
		\ket{x_{j+1}},
	\end{eqnarray}
	where the last equality is an identity which simply follows from the explicit definition of the maps $\mathcal{P}_x$. This observation is at the core of the relation between Markovianity and regression, which, while well-known, shall be recalled explicitly for the sake of completeness:
	\begin{proposition}\label{prop:general}
		The following statements are equivalent:
		\begin{itemize}
			\item[(i)] the regression equality~\eqref{eq:probquantum_regression} holds for all $n$;
			\item [(ii)] the process is Markovian, and Eq.~\eqref{eq:probquantum_regression} holds for $n=1,2$;
			\item [(iii)] the process is Markovian, and satisfies
			\begin{eqnarray}
				\mathbb{P}_1(x_1,t_1)&=&\bra{x_1}\Lambda_{t_1}(\rho)\ket{x_1};\\
				\mathbb{P}(x_2,t_2|x_1,t_1)&=&\bra{x_2}\Lambda_{t_2,t_1}(\ketbra{x_1}{x_1})\ket{x_2}.
			\end{eqnarray}
		\end{itemize}
	\end{proposition}
	\begin{proof}
		The equivalence (ii)$\iff$(iii) is an immediate consequence of Eq.~\eqref{eq:probquantum_regression}. (i)$\implies$(ii) is an immediate consequence of the fact that, if (i) holds, then the full family of joint probabilities is given by Eq.~\eqref{eq:probquantum_regression} and thus
		\begin{equation}
			\mathbb{P}(x_n,t_n|x_{n-1},t_{n-1};\ldots;x_1,t_1)=\bra{x_n}\Lambda_{t_n,t_{n-1}}\left(\ketbra{x_{n-1}}{x_{n-1}}\right)\ket{x_n}=\mathbb{P}(x_n,t_n|x_{n-1},t_{n-1}).
		\end{equation}
		Vice versa, suppose that (ii) is true. Since the process is Markovian, the family of joint probabilities can be entirely reconstructed as such:
		\begin{equation}\label{eq:reconstruct1}
			\mathbb{P}_n(x_n,t_n;\ldots;x_1,t_1)=\mathbb{P}(x_n,t_n|x_{n-1},t_{n-1})\cdots\mathbb{P}(x_2,t_2|x_1,t_1)\mathbb{P}_1(x_1,t_1).
		\end{equation}
		By Eq.~\eqref{eq:reconstruct1} and the fact that Eq.~\eqref{eq:probquantum_regression} holds for $n=1,2$, one easily concludes that the latter equality holds for all $n$.
	\end{proof}
	
	This shows the link between the validity of the regression formula \textit{for a fixed orthonormal basis}, and the Markovianity in the same basis: $n$-point regression implies Markovianity, and Markovianity plus $1$-point and $2$-point regression implies $n$-point regression.
		
	\subsection{
		\texorpdfstring{Measurements in bases compatible with the $\mathcal{H}_\g \oplus \mathcal{H}_\e$ splitting}{Measurements in the basis compatible with the orthogonal splitting}\label{subsec:compatible}
	}
	
	Consider a projective measurements w.r.t. an orthonormal basis \textit{compatible} with the splitting of the Hilbert space $\hilb_{\rm S}$ of the system into its excited and ground sector, that is, a basis $\{\ket{x}\}_{x=0,\dots,d}$ such that $\ket{0}$ is the ground state of the atom and $\{\ket{x}\}_{x=1,\dots,d}\subset\hilb_{\e}$ is any orthonormal basis of the excited sector. Let us compute the corresponding $n$-point joint probability. Assuming the following initial state $|\Psi_0\> = |\psi\> \otimes \vac$, with $|\psi\> = \alpha |0\> \oplus |\psi_\e\>$, and using hereafter the compact notation
	\begin{equation}
		b^\dag_j\left(\eta\right)=\int\mathrm{d}\omega\;\eta(\omega)b_j^\dag(\omega)
	\end{equation}
	for an arbitrary wavefunction $\eta$, one finds, as discussed in Section~\ref{sec:flat},
	\begin{equation}\label{Psi-t}
		U_{t_1}|\Psi_0\> = \Big[ \alpha |0\> \otimes |\vac\> + \A(t_1)|\psi_\e\> \Big] \otimes \ket{\vac} + |0\> \otimes
		\sum_{j=1}^d  b_j^\dagger(\xi_j(t_1))|\vac\>  ,
	\end{equation}
	with $\A(t)$ as in Eq.~\eqref{eq:exp} and $\xi_j(t)$ as in Eq.~\eqref{xi-A}; hence,
	\begin{eqnarray}
		P_{x_1} \otimes \oper_B\, U_{t_1}|\Psi_0\> &=&  \<x_1|\A(t_1)|\psi_\e\>   \, |x_1\> \otimes |\vac\> , \\
		P_0 \otimes \oper_B\, U_{t_1}|\Psi_0\> &=&   |0\> \otimes \Big(\alpha |\vac\> +  \sum_{j=1}^{d}   b_j^\dagger(\xi_j(t_1))|\vac\>\Big) .
	\end{eqnarray}
	Therefore, we have	
	\begin{equation}\label{eq:p1}
		\mathbb{P}_1(x_1,t) = | \<x_1|\A(t)|\psi_\e\> |^2 \ , \ \ \  \mathbb{P}_1(0,t) = |\alpha|^2 +  \sum_{j=1}^{d}  \|\xi_j(t)\|^2 .
	\end{equation}
	Now we must apply the second unitary evolution $U_{\Delta t_2}$. We have
	\begin{eqnarray}
		U_{\Dt_2} P_{x_1} \otimes \oper_B\, U_{t_1}|\Psi_0\> &=&\bra{x_1}\A(t_1)\ket{x_1}\,U_{\Dt_2}\left(\ket{x_1}\otimes\ket{\vac}\right)  \\
		&=&  \<x_1|\A(t)|\psi_\e\>  \Big( \A(\Dt_2) |x_1\> \otimes |\vac\> + |0\> \otimes
		\sum_{j=1}^d  b_j^\dagger(\eta^{(1)}_j(\Dt_2))|\vac\> \Big) \nonumber ,
	\end{eqnarray}
	where $\Dt_k = t_k - t_{k-1}$, with the function $\eta^{(\alpha)}_j(\tau)$ being defined as follows:
	\begin{equation}\label{eta}
		\eta^{(\alpha)}_j(\tau,\omega) = -i\int_0^\tau \mathrm{d}s\;\e^{-\i\omega(\tau-s)}  f_j(\omega) \< \e_j|\A(s)|x_\alpha\>  ,
	\end{equation}
	and we used again the explicit form of the evolution in the single-excitation sector. Besides, using Prop.~\ref{PRO-A},
	\begin{eqnarray}\label{x}
		U_{\Dt_2} P_0 \otimes \oper_B\, U_{t_1}|\Psi_0\> &=&U_{\Dt_2}\left[ |0\> \otimes \Big(\alpha |\vac\> +  \sum_{j=1}^{d}   b_j^\dagger(\xi_j(t_1))|\vac\>\Big) \right]\nonumber\\&=& |0\> \otimes \Big(\alpha |\vac\> +  \sum_{j=1}^{d}  b_j^\dagger(\xi_j(\Dt_2,t_1))|\vac\>\Big)  ,
	\end{eqnarray}
	where $	\xi_j(\tau,t,\omega) = \e^{-\i\omega \tau} \xi_j(t,\omega)$. Notice that, differently from all previous steps, this is the \textit{only} step that holds specifically for the flat coupling case.
	
	The second projective measurement gives rise to	
	\begin{eqnarray}
		P_{x_2} \otimes \oper_B\,  U_{\Dt_2} P_{x_1} \otimes \oper_B\, U_{t_1}|\Psi_0\> &=& \<x_2|\A(\Dt_2)|x_1\>\!\<x_1|\A(t_1)|\psi_\e\> \, |x_2\> \otimes |\vac\> , \\
		P_{0} \otimes \oper_B\,  U_{\Dt_2} P_{x_1} \otimes \oper_B\, U_{t_1}|\Psi_0\> &=&  \<x_1|\A(t_1)|\psi_\e\>  \, |0\> \otimes \sum_{j=1}^d  b_j^\dagger(\eta^{(1)}(\Dt_2))|\vac\> , \\
		P_{x_2} \otimes \oper_B\,  U_{\Dt_2} P_{0} \otimes \oper_B\, U_{t_1}|\Psi_0\> &=& 0 , \\
		P_{0} \otimes \oper_B\,  U_{\Dt_2} P_{0} \otimes \oper_B\, U_{t_1}|\Psi_0\> &=& |0\> \otimes \Big(\alpha |\vac\> +  \sum_{k=1}^{d}  b_j^\dagger(\xi_j(\Dt_2,t_1))|\vac\>\Big) .
	\end{eqnarray}
	One finds
	\begin{eqnarray}
		\mathbb{P}_2(x_2,t_2;x_1,t_1)    &=& |\<x_2|\A(\Dt_2)|x_1\>|^2 \, \mathbb{P}_1(x_1,t_1) \\
		\mathbb{P}_2(0,t_2;x_1,t_1)   &=& \| \eta^{(1)}_{x_1}(\Dt_2)\|^2 \, \mathbb{P}_1(x_1,t_1) \\
		\mathbb{P}_2(x_2,t_2;0,t_1)  &=& 0 , \\
		\mathbb{P}_2(0,t_2;0,t_1)  &=&  \mathbb{P}_1(0,t_1) .
	\end{eqnarray}
	It is straightforward to generalize the above procedure for $n$-point joint probabilities:\small
	\begin{eqnarray}
		\mathbb{P}_n(x_n,t_n;x_{n-1},t_{n-1};\ldots ;x_1,t_1) &=& |\<x_{n}|\A(\Dt_n)|x_{n-1}\>|^2 \ldots |\<x_2|\A(\Dt_2)|x_1\>|^2 \, \mathbb{P}_1(x_1,t_1) , \label{P1} \\
		\mathbb{P}_n(0,t_n;\ldots;0,t_{k+1};x_{k},t_{k};\ldots ;x_1,t_1) &=& \|\eta^{(k)}_{x_k}(\Dt_{k+1})\|^2 \,  \mathbb{P}_k(x_k,t_k;x_{k-1},t_{k-1};\ldots ;x_1,t_1) , \label{P2}  \\
		\mathbb{P}_n(0,t_n;0,t_{n-1};\ldots ;0,t_1) &=& \mathbb{P}_1(0,t_1) , \label{P3}
	\end{eqnarray}\normalsize
	and all remaining probabilities vanish. Importantly, all joint probabilities associated to phenomena in which the system is \textit{first} measured in its ground state ($x_j=0$) and \textit{later} in an excited state ($x_{j+1}\neq0$) vanish: physically, because of our particular choice of form factors, emitted photons cannot be reabsorbed---the decay of the system is irreversible.
	
	\subsubsection{Markovianity}
	Having computed the full family of joint probabilities associated with the system, we can now investigate its Markovianity. A direct scrutiny of Eqs.~\eqref{P1}--\eqref{P3} shows that the process is indeed Markovian:
	
	\begin{proposition}
		Let $\textbf{H}$ be a multilevel spin--boson model in the limit of flat coupling, $\ket{\psi_0}=\ket{\psi_\e}+\alpha\ket{0}\in\hilb_{\rm S}$, and $\{\ket{x}\}_{x=1,\dots,d}\subset\hilb_{\rm S}$ any orthonormal basis of the excited sector $\hilb_{\e}$. Then the process defined by Eq.~\eqref{I} is a Markovian process with $1$-time probability $\mathbb{P}_1(x,t)$ and transition probabilities $\mathbb{P}(x,t|y,s)$ given as follows: for any $x=1,\dots,d$ and $t\geq s\geq 0$,
		\begin{eqnarray}\label{eq:p1-1}
			\mathbb{P}_1(x,t)&=&\left|\braket{x|\A(t)|\psi_\e}\right|^2,\\\label{eq:p1-2}
			\mathbb{P}_1(0,t)&=&1-\sum_{x=1}^d\left|\braket{x|\A(t)|\psi_\e}\right|^2
		\end{eqnarray}
		and
		\begin{eqnarray}
			\mathbb{P}(x,t|y,s)  &=&   | \<x|\A(t-s)|y\> |^2  , \label{p11} \\
			\mathbb{P}(0,t|y,s)  &=&  1-\sum_{x=1}^d | \<x|\A(t-s)|y\> |^2\label{p01}, \\
			\mathbb{P}(x,t|0,s)  &=& 0 , \label{p10} \\
			\mathbb{P}(0,t|0,s)  &=& 1 . \label{p00}
		\end{eqnarray}
		Finally, the process satisfies the regression equality~\eqref{eq:probquantum_regression}.
	\end{proposition}
	\begin{proof}
		Eqs.~\eqref{eq:p1-1}--\eqref{eq:p1-2} simply coincide with Eq.~\eqref{eq:p1} taking into account the unitarity of the evolution induced by $\textbf{H}$; similarly, Eqs.~\eqref{p11}--\eqref{p00} follow from Eq.~\eqref{P1}--\eqref{P3} and, again, the unitarity of the evolution. Moreover, recalling the definition of the reduced dynamics $\Lambda_t$ and the corresponding propagator, we indeed have, for all $x,y=0,1,\dots,d$,
		\begin{equation}
			\mathbb{P}(x,t|y,s) =\< x| \Lambda_{t,s}(|y\>\!\<y|)| x\>  ,
		\end{equation}
		whence the final claim immediately follows from Prop.~\ref{prop:general}.
	\end{proof}
	
	It is important to notice that these properties hold regardless of whether the orthonormal basis $\{\ket{x}\}_{x=1,\dots,d}\subset\hilb_{\e}$ of the excited sector coincides with the ``natural" basis $\{\ket{\e_x}\}_{x=1,\dots,d}$ in which the operator $\Gamma$ is diagonal---what is really important is the fact that, choosing an orthonormal basis of the excited sector as our measurement basis, we are forced to "complete" it with the ground state $\ket{0}$ of the system, thus obtaining an orthonormal basis of $\hilb_{\rm S}$ which is compatible with the excited-ground splitting. The interpretation of Eqs.~\eqref{p11}--\eqref{p00} is also immediate: if the system is found in its ground state $\ket{0}$ after some measurement, then it will be found in such a state in all future measurements. This is ultimately a consequence of the fact that, with our choice of measurement basis, the global state of the system+bath is never going to leave the excited sector. Also notice that, by Eq.~\eqref{p11}, the transition probability $\mathbb{P}(x,t|y,s)$ for $x,y=1,\dots,d$ (that is, the probability associated with a transition $\ket{y}\rightarrow\ket{x}$ between two states of the excited sector) simply coincides with the squared modulus of the matrix element of $\A(t-s)$ between such states.
	
	\subsubsection{Classicality}
	Let us now examine whether the process is classical, i.e., whether it satisfies the consistency conditions~\eqref{eq:consistency}; since the process is Markovian, this happens if and only if the Chapman--Kolmogorov equations~\eqref{eq:chapman0}--\eqref{eq:chapman} for the transition probability is satisfies by the conditional probabilities~\eqref{p11}--\eqref{p00}. For the multilevel spin--boson model it means that for all $x,y=1,\dots,d$ and $t\geq r\geq s\geq 0$,
	\begin{eqnarray}
		\label{eq:ck-xy}	\mathbb{P}(x,t|y,s)&=&\sum_{z=0}^d\mathbb{P}(x,t|z,r)\mathbb{P}(z,r|y,s);\\
		\label{eq:ck-0y}	\mathbb{P}(0,t|y,s)&=&\sum_{z=0}^d\mathbb{P}(0,t|z,r)\mathbb{P}(z,r|y,s);\\
		\label{eq:ck-x0}	\mathbb{P}(x,t|0,s)&=&\sum_{z=0}^d\mathbb{P}(x,t|z,r)\mathbb{P}(z,r|0,s);\\
		\label{eq:ck-00}	\mathbb{P}(0,t|0,s)&=&\sum_{z=0}^d\mathbb{P}(0,t|z,r)\mathbb{P}(z,r|0,s).
	\end{eqnarray}
	By Eqs.~\eqref{p11}--\eqref{p00}, Eqs.~\eqref{eq:ck-x0}--\eqref{eq:ck-00} reduce to the identities $0=0$ and $1=1$, but Eqs.~\eqref{eq:ck-xy}--\eqref{eq:ck-0y} are not trivial and are not generally satisfied.
	
	As it turns out, classicality \textit{does} depend on the choice of measurement basis.	
	
	\begin{proposition}  \label{PRO-D}
		Consider the same process as before with the additional condition $[\H_\e,\Gamma]=0$. Then:
		\begin{itemize}
			\item if $\ket{x}=\ket{\e_x}$ for all $x=1,\dots,d$, then the process is classical;
			\item suppose that the spectrum of $\Gamma$ is nondegenerate, i.e. $\gamma_j\neq\gamma_\ell$ for $j\neq\ell$. Then the process is classical \textit{if and only if} the chosen measurement basis is such that, for all $x=1,\dots,d$, $\ket{x}=\ket{\e_x}$ possibly up to a phase shift.
		\end{itemize}
	\end{proposition}
	
	The proof of this statement is reported in Appendix~\ref{app:1a}; we shall discuss here the meaning of this statement. When $[H_\e,\Gamma]=0$, the basis $\{\ket{\e_j}\}_{j=1,\dots,d}$ of eigenvectors of the decay operator $\Gamma$ is also an eigenbasis for $H_\e$. Consequently, as discussed in the previous section, the evolution of the system does not ``mix" the various eigenstates nor the modes of the boson field, whence the system effectively behaves as a two-level system interacting with a single bosonic mode. In this situation, a simple computation shows that both Chapman--Kolmogorov equations are satisfied. The less trivial part of the proof of the statement above is to show that, indeed, this is the \textit{only} scenario in which classicality holds: any other measurement basis will cause nonclassical effects to emerge.
	
	It is instructive to investigate this phenomenon at the light of the results of~\cite{Andrea-1,Andrea-2} about the classicality of quantum Markov processes, which we briefly summarized in the Introduction. To this purpose, let us recall their results in a more precise way. A family of propagators is said to be non-coherence-generating-and-detecting (NCGD) with respect to a basis $\{\ket{x}\}_{x=0,\dots,d}$ of the system Hilbert space if it satisfies the following equality:			
	\begin{equation}\label{DDD}
		\Delta \circ \Lambda_{t_{k+1},t_k} \circ \Delta \circ \Lambda_{t_{k},t_{k-1}} \circ \Delta = \Delta  \circ \Lambda_{t_{k+1},t_{k-1}} \circ \Delta ,
	\end{equation}
	where
	\begin{equation}
		\Delta = \sum_{x=0}^d \mathcal{P}_x ,\qquad \mathcal{P}_x=\ketbra{x}{x}\cdot\ketbra{x}{x},
	\end{equation}
	is a quantum channel representing perfect decoherence w.r.t. $\{|0\>,|x\>_{x=1}^d\}$. Eq.~\eqref{DDD} essentially means that, at any intermediate time $t_k$ between any couple of measurements at two times $t_{k+1}$ and $t_{k-1}$, not doing any measure---that is, applying the identity map---is indistinguishable from performing a measurement and then averaging on all results---that is, applying $\Delta$. In its essence, this condition may be regarded as the transposition of the requirement~\eqref{eq:consistency} from the level of multitime statistics to the underlying quantum channels which realize the statistics; this is essentially the content of the following characterization:
	
	\begin{theorem}[\!\!\cite{Andrea-2}, Theorem 1]
		Consider a family of joint probabilities $\{\mathbb{P}_n\}_{n}$ satisfying the Markov property. Then the family is classical if and only if there exist, on a suitable Hilbert space $\hilb$,
		\begin{itemize}
			\item a basis $\{\ket{x}\}_{x=0,\dots,d}$;
			\item a state $\tilde{\rho}_{0}\in\mathcal{B}(\hilb)$ which is diagonal in said basis;
			\item a family of quantum channels $\tilde{\Lambda}_{t_{j+1},t_j}$ which are NCGD with respect to said basis,
		\end{itemize}
		such that
		\begin{equation}
			\mathbb{P}_n(x_n,t_n;\dots;x_1,t_1)=\tr\left[\mathcal{P}_{x_n}\tilde{\Lambda}_{t_n,t_{n-1}}\cdots\mathcal{P}_{x_1}\tilde{\Lambda}_{t_1}(\tilde{\rho}_{0})\right].
		\end{equation}
	\end{theorem}
	
	\textit{Apparently}, the theorem above may seem in contrast with the content of Prop.~\ref{PRO-D}, since the initial state assumed in our model, $\rho_{0}=\ketbra{\psi_\e}{\psi_\e}$, is generally not diagonal in the measurement basis. In fact, this is not the case. As thoroughly discussed in~\cite{Andrea-2}, the maps $\tilde{\Lambda}_{t,s}$ and the initial state $\rho_{0}$ define an \textit{artificial} dynamics of the system which, for a multitime statistics obtained by a preexisting underlying quantum system, does \textit{not} generally coincide with the actual reduced dynamics nor the actual initial state, while reproducing the same statistics.
	
	This can be better understood by looking at the explicit construction of $\tilde{\Lambda}_{t_{j+1},t_j}$ and $\tilde{\rho}_{0}$ for a given multitime statistics provided in~\cite{Andrea-2}. Setting $t_1=0$, they define the channels $\tilde{\Lambda}_{t_{j+1},t_j}$ and the initial state $\tilde{\rho}_{0}$ via
	\begin{eqnarray}
		\tilde{\Lambda}_{t_{j+1},t_j}\left(\ketbra{x_j}{y_j}\right)&=&\delta_{x_j,y_j}\sum_{x_{j+1}=0}^d\mathbb{P}\left(x_{j+1},t_{j+1}|x_j,t_j\right)\ketbra{x_j}{x_j};\\
		\tilde{\rho}_{0}&=&\sum_{x_1=0}^d\mathbb{P}_1(x_1,t_1)\ketbra{x_1}{x_1},
	\end{eqnarray}
	and clearly, in our case, neither $\tilde{\Lambda}_{t_{j+1},t_j}$ nor $\tilde{\rho}_{0}$ coincide with the actual reduced dynamics of the model nor the initial state. However, it is worth noticing that the actual propagators $\Lambda_{t_j,t_{j+1}}$ are indeed NCDG in our case.
	
	Summing up: the multilevel spin--boson model $\textbf{H}$, repeatedly probed via sharp measurements associated via an orthonormal basis compatible with the $\hilb_{\e}\oplus\hilb_\g$ splitting, is Markovian---in the sense adopted in this paper---in the limit of flat couplings. Besides, while generally non-classical, the process becomes classical if $[\H_\e,\Gamma]=0$ \textit{and} the chosen measurement basis coincides with the common basis of eigenvectors of $\H_\e$ and $\Gamma$. Roughly speaking, this happens because, when adopting this particular basis, no probability exchange between states competing with different eigenvectors $\ket{\e_j}$ happens---no fundamentally quantum feature is unveiled in the measurement. Choosing any other basis will reveal the nonclassicality of the process.
	
	\subsection{Measurements in arbitrary bases}\label{subsec:arbitrary}	
	
	Consider now a projective measurements w.r.t. to an arbitrary orthonormal basis $\{|x_\alpha\>\}_{\alpha=0,\dots,d}$ in $\mathcal{H}$. With this choice, it is no longer true that the entire process is constrained in the single-excitation sector: in fact, after the $n$th measurement, the following vector
	
	\begin{equation}
		|\Psi^{x_n,\ldots,x_1}_{t_n,\ldots,t_1}\> = P_{x_n} \otimes \oper_B U_{\Dt_{n}} \ldots U_{\Dt_2} P_{x_1} \otimes \oper_B U_{t_1} |\Psi_0\> ,
	\end{equation}
	belongs to the $(n+1)$-excitation sector. This fact makes the analysis quite technical and we present all the necessary details in the Appendix. Evidently, the process is not classical; however, it is still Markovian.
	
	\begin{theorem} \label{THM}
		The process defined by		
		\begin{equation}
			\mathbb{P}_n(x_n,t_n;\ldots ;x_1,t_1) = \|  \, |\Psi^{x_n,\ldots,x_1}_{t_n,\ldots,t_1}\>\, \|^2
		\end{equation}
		satisfies the regression equality~\eqref{eq:probquantum_regression} for an arbitrary orthonormal basis $\{|x_\alpha\>\}_{\alpha=0,\dots,d}\subset\hilb_\e\oplus\hilb_\g$.
	\end{theorem}
	
	The proof, reported in the Appendix, is based on the following observation: any basis vector $|x_\alpha\>$ may be uniquely decomposed as $|x_\alpha\> = \mu_\alpha |0\> \oplus |\tilde{x}_\alpha\>$ with $|\tilde{x}_\alpha\> \in \mathcal{H}_\e$. One has		
	\begin{equation}
		U_{t}|\tilde{x}_\alpha\> \otimes |\vac\> =  \A(t)|\tilde{x}_\alpha \>  \otimes |\vac\> + |0\> \otimes \sum_{j=1}^d  b_j^\dagger(\tilde{\eta}^{(\alpha)}_j(t))|\vac\>  ,
	\end{equation}
	where the function $\eta^{(\alpha)}_j(\tau)$ is defined via	
	\begin{equation}\label{eta-t}
		\tilde{\eta}^{(\alpha)}_j(t,\omega) = -i\int_0^\tau \mathrm{d}\omega\;\e^{-\i\omega(t-s)}  f_j(\omega) \< \e_j|\A(s)|\tilde{x}_\alpha\>  , \ \ \ j=1,\ldots,d .
	\end{equation}
	Similarly, let us define	
	\begin{equation}
		\tilde{\eta}^{(\alpha)}_j(\tau,t,\omega) = \e^{-\i\omega \tau} \tilde{\eta}^{(\alpha)}_j(t,\omega) .
	\end{equation}
	With these definitions, the proof is based on two technical properties that hold specifically in the flat coupling limit, the first being Prop.~\ref{PRO-A}, and the second one, again shown in the Appendix, being the following one:
	\begin{proposition} \label{PRO-bbb} The following property holds:		
		\begin{eqnarray}   \label{bbb}
			&& \left\|\,  b^\dagger_{j_1}(\xi_{j_1}(t_{n+1}-t_1,\Dt_1)) \, \prod_{k=2}^n
			\sum_{j_k=1}^d b^\dagger_{j_k}(\tilde{\eta}_{j_k}^{(k-1)}(t_{n+1}-t_k,\Dt_k)) |\vac\> \right\|^2 \nonumber \\
			&& = \left\|  b^\dagger_{j_1}(\xi_{j_1}(\Dt_1))\right\|^2  \, \prod_{k=2}^n
			\left\|
			\sum_{j_k=1}^d b^\dagger_{j_k}(\tilde{\eta}_{j_k}^{(k-1)}(t_{n+1}-t_k,\Dt_k)) \right\|^2  ,
		\end{eqnarray}
		for any $n\geq 1$.
	\end{proposition}	
	The above property allows to prove the regression equality~\eqref{eq:probquantum_regression} (cf.\ Appendix) and thus Markovianity.
	
	\begin{remark} \label{RE-a}
		Actually, the proof of  Theorem~\ref{THM} shows that the regression equality~\eqref{eq:probquantum_regression} holds for a much more general measurement strategy, that is, without using the same fixed orthonormal basis to measure the system at $\{t_1,t_2,\ldots\}$, but rather using at each moment $t_k$ any arbitrary orthonormal basis $\mathcal{B}_k$. Then $\mathbb{P}_n(x_n,t_n;\ldots;x_1,t_1)$ defines the joint probability of obtaining $\{x_n,\ldots,x_1\}$ probing the system at times $\{t_n,\ldots,t_1\}$ w.r.t. with (arbitrary) orthonormal bases $\{\mathcal{B}_n,\ldots,\mathcal{B}_1\}$.
	\end{remark}
	
	\subsection{Past-future independence}
	
	We shall conclude the work with the following observation. Recently, Budini and collaborators~\cite{Budini1,Budini2,Budini3} introduced an interesting non-Markovianity witness based on the so-called \textit{conditional past-future} (CPF) correlations. Essentially, the vanishing of CPF correlations is equivalent to the validity of the regression equality~\eqref{eq:probquantum_regression} for  projective measurements w.r.t. arbitrary orthonormal basis $\{\mathcal{B}_n,\ldots,\mathcal{B}_1\}$ (cf.\ Remark~\ref{RE-a}).
	
	In this regard, consider a normalized state vector	
	\begin{equation}
		|\widetilde{\Psi}^{x_n,\ldots,x_1}_{t_n,\ldots,t_1}\> := \frac{| \Psi^{x_n,\ldots,x_1}_{t_n,\ldots,t_1} \>}{\sqrt{\mathbb{P}_n(x_n,t_n;\ldots;x_1,t_1)}} = |x_n\> \otimes |\varphi^{x_n,\ldots,x_1}_{t_n,\ldots,t_1}\> ,
	\end{equation}
	with $|\varphi^{x_n,\ldots,x_1}_{t_n,\ldots,t_1}\>$ corresponding to the boson state vector which depends on the history of all $n$ measurements up to time $t_n$. Now, the reduced dynamics starting at $t=t_n$	
	\begin{equation}
		|x_n \>\!\<x_n| \to   {\rm Tr}_B \Big( U_\tau |x_n\>\!\<x_n| \otimes |\varphi^{x_n,\ldots,x_1}_{t_n,\ldots,t_1}\>\!\< \varphi^{x_n,\ldots,x_1}_{t_n,\ldots,t_1}| U^\dagger_\tau \Big)
	\end{equation}
	provides a dynamical map which depends on the boson state vector $|\varphi^{x_n,\ldots,x_1}_{t_n,\ldots,t_1}\>$, and hence it depends upon the entire measurement history up to time $t_n$. The process is Markovian whenever the above dynamical map \textit{does not depend} upon the history of measurements records. For the multilevel spin-boson model we are considering, one has
	\begin{equation}
		{\rm Tr}_B \Big( U_\tau |x_n\>\!\<x_n| \otimes |\varphi^{x_n,\ldots,x_1}_{t_n,\ldots,t_1}\>\!\< \varphi^{x_n,\ldots,x_1}_{t_n,\ldots,t_1}| U^\dagger_\tau \Big)
		= {\rm Tr}_B \Big( U_\tau |x_n\>\!\<x_n| \otimes |\vac\>\!\<\vac| U_\tau^\dagger \Big) = \Lambda_\tau( |x_n\>\!\<x_n|) ,
	\end{equation}
	that is, one may erase the memory about measurement records replacing $|\varphi^{x_n,\ldots,x_1}_{t_n,\ldots,t_1}\>$ by the vacuum state $|\vac\>$.
	
	\section{Conclusions}\label{sec:conclusions}
	
	We have provided a detailed analysis of Markovianity and classicality for the multitime statistics associated with a paradigmatic model describing a multilevel system coupled with a multimode boson field, for which calculations can carried out explicitly, probed in an arbitrary fixed basis at different times. In the limit in which the coupling between the system and the environment is described by flat functions---that is, all field frequencies are coupled with the same strength---this system was proven to be Markovian for any measurement basis and any choice of the initial (pure) system--bath state. Furthermore, in the same limit, the model was proven to be classical if and only if the selected measurement basis coincides with the one (if any) that diagonalizes the reduced evolution of the excited sector; any other choice of basis will involve probability exchange between states competing with different eigenvectors, hence reveling the underlying quantum nature of the system. This result was compared with the characterization of classicality for quantum Markov processes in terms of the generation of coherence in the process, provided in~\cite{Andrea-1,Andrea-2}.
	
	Our work provides a detailed exposition of fundamental concepts of growing interest in recent years---Markovianity and classicality---applied to particular instances of quantum systems having a relatively simple structure, while nevertheless retaining a sufficiently rich phenomenology of possible cases. In the context of our proposed model, the slightly counterintuitive fact that quantum systems \textit{can} behave classically, when probed in a proper way, is easily interpreted in terms of the presence or absence of mixing between eigenvectors.
	
	Among many possible future developments of our work, it would be useful to provide a similar discussion of the classicality of our model beyond the flat-coupling limit---that is, outside the Markovian scenario. In this case, classicality can be characterized in terms of the generation of quantum discord rather than coherence~\cite{Andrea-2}. Intuitively, it would be tempting to conjecture that the model under investigation in the present work will behave classically when probed via the same measurement basis under which it is classical in the Markov regime: indeed, the reduced dynamics in the excited sector will be again diagonal in such a basis. We will leave a detailed answer of this question to future works. Finally, as a natural continuation of this line of research, \textit{quantifying} nonclassicality in spin--boson models and their generalizations, thus going forward the mere characterization of processes that happen to be classical, would be of paramount importance.
	
	\section*{Acknowledgments}
	
	DC was supported by the Polish National Science Center project No. 2018/30/A/ST2/00837. DL was partially supported by Istituto Nazionale di Fisica Nucleare (INFN) through the project “QUANTUM” and by the Italian National Group of Mathematical Physics (GNFM--INdAM), and acknowledges support by MIUR via PRIN 2017 (Progetto di Ricerca di Interesse Nazionale), project QUSHIP (2017SRNBRK); he also thanks the Institute of Physics at the Nicolaus University in Toru\'n for its hospitality.
	
	\appendix
	
	\section{Proof of Prop.~\ref{PRO-A}}\label{app:1}
	
	Following~\cite{hidden2} we prove the following result: given $\eta(\omega)$  let us consider the  corresponding Fourier (anti--)transform
	\begin{equation}
		\hat{\eta}(x) :=\int\mathrm{d}\omega\;\e^{\i\omega x}\eta(\omega) .
	\end{equation}
	
	\begin{lemma}
		Let $\hat{\eta}(x)=0$ for $x\in [-T,0]$ (with $T>0$). Then		
		\begin{equation}
			\e^{-\i t\textbf{H}}\oper\otimes b_\alpha^\dag(\eta)\e^{\i t\textbf{H}}=\oper\otimes b_\alpha^\dag\left(\e^{-\i t\omega}\eta\right)\qquad \forall t\in[0,T],
		\end{equation}
		or equivalently, in the position representation,
		\begin{equation}
			\e^{-\i t\textbf{H}}\oper\otimes b_\alpha^\dag(\hat\eta)\e^{\i t\textbf{H}}=\oper\otimes b_\alpha^\dag\left(\hat\eta(\cdot-t)\right)\qquad \forall t\in[0,T].
		\end{equation}
	\end{lemma}
	\begin{proof}
		Let us define
		\begin{equation}
			B^\dag_\alpha(\eta;t):=\e^{-\i t\textbf{H}}\oper_S \otimes b_\alpha^\dag(\eta)\e^{\i t\textbf{H}},
		\end{equation}
		and note that this is the unique solution of the Heisenberg equation
		\begin{equation}\label{eq:werner}
			\begin{cases}
				\frac{\mathrm{d}}{\mathrm{d}t}B^\dag_\alpha(\eta;t)=-\i\left[\textbf{H},B^\dag_\alpha(\eta;t)\right],\\
				B^\dag_\alpha(\eta;0)=\oper \otimes b^\dag_\alpha(\eta).
			\end{cases}
		\end{equation}
		We shall verify that
		\begin{equation}\label{eq:ansatz}
			t\mapsto\oper\otimes b^\dag_\alpha(\e^{-\i t\omega}\eta)
		\end{equation}
		solves Eq.~\eqref{eq:werner} provided $\hat{\eta}(x)=0$ for $x\in [-T,0]$. Indeed, we have
		\begin{eqnarray}
			\frac{\mathrm{d}}{\mathrm{d}t}\left(\oper_S\otimes b^\dag_\alpha(\e^{-\i t\omega}\eta)\right)&=&\oper_S\otimes \frac{\mathrm{d}}{\mathrm{d}t}b^\dag_\alpha(\e^{-\i t\omega}\eta)\nonumber\\
			&=&-\i\oper_S\otimes[H_{\rm B},b^\dag_\alpha(\e^{-\i t\omega}\eta)],
		\end{eqnarray}
		whereas
		\begin{eqnarray}
			[\textbf{H},\oper\otimes b^\dag_\alpha(\e^{-\i t\omega}\eta)]&=&\oper\otimes[H_{\rm B},b^\dag_\alpha(\e^{-\i\omega t}\eta)]+\ketbra{\e_\alpha}{\g}\otimes[b_\alpha(f_\alpha),b^\dag_\alpha(\e^{-\i\omega t}\eta)].
		\end{eqnarray}
		But
		\begin{eqnarray}
			[b_\alpha(f_\alpha),b^\dag_\alpha(\e^{-\i\omega t}\eta)]&=&\int\mathrm{d}\omega\;f_\alpha(\omega)^*\e^{-\i\omega t}\eta(\omega)\nonumber\\
			&=&\sqrt{\gamma_\alpha}\int\mathrm{d}\omega\;\e^{-\i\omega t}\eta(\omega)\nonumber\\
			&=&\sqrt{\gamma_\alpha}\,\hat{\eta}(-t).
		\end{eqnarray}
		Therefore, under our assumptions, the desired equality holds for all times $t\in[0,T]$. In particular, if $\eta(x)=0$ for all $x\leq0$, the desired equality holds for all times $t\geq0$.	
	\end{proof}
	
	\section{Proof of Prop.~\ref{PRO-D}}\label{app:1a}
	
	Let us consider the particular case in which $[\H_\e,\Gamma]=0$, that is, both $\H_\e$ and $\Gamma$ are diagonal in the basis $\{\ket{\e_x}\}_{x=1,\dots,d}$. In this case, $\A(t)$ is diagonal as well in this basis, that is, by Eq.~\eqref{eq:at-diag}, choosing $\ket{x}=\ket{\e_x}$ for all $x=1,\dots,d$, the $1$-time probabilities become
	\begin{eqnarray}
		\mathbb{P}_1(x,t)&=&\e^{-\gamma_x t}|\braket{x|\psi_\e}|^2;\\
		\mathbb{P}_1(0,t)&=&1-\sum_{x=1}^d\e^{-\gamma_x t}|\braket{x|\psi_\e}|^2,
	\end{eqnarray}
	and the transition probabilities~\eqref{p11}--\eqref{p01} simplify as
	\begin{eqnarray}
		\mathbb{P}(x,t|y,s)  &=&   \e^{-\gamma_x(t-s)}\,\delta_{xy}  , \label{p11-diag} \\
		\mathbb{P}(0,t|y,s)  &=&  1-\e^{-\gamma_y(t-s)}\,\label{p01-diag},
	\end{eqnarray}
	and all Chapman--Kolmogorov equations are easily satisfied. If, instead, $\{\ket{x}\}_{x=1,\dots,d}$ is a generic orthonormal basis of $\hilb_{\e}$, then, again by Eq.~\eqref{eq:at-diag},
	\begin{equation}
		\left|\braket{x|\A(t-s)|y}\right|^2=\sum_{j=1}^d\e^{-\gamma_j(t-s)}|\!\braket{x|\e_j}\!|^2\,|\!\braket{y|\e_j}\!|^2,
	\end{equation}
	whence the transition probabilities~\eqref{p11}--\eqref{p01} read
	\begin{eqnarray}
		\mathbb{P}(x,t|y,s)  &=&   \sum_{j=1}^d\e^{-\gamma_j(t-s)}|\!\braket{x|\e_j}\!|^2\,|\!\braket{y|\e_j}\!|^2 , \label{p11-ndiag} \\
		\mathbb{P}(0,t|y,s)  &=&  1-\sum_{x=1}^d\sum_{j=1}^d\e^{-\gamma_j(t-s)}|\!\braket{x|\e_j}\!|^2\,|\!\braket{y|\e_j}\!|^2\nonumber\\
		&=&1-\sum_{j=1}^d\e^{-\gamma_j(t-s)}|\!\braket{y|\e_j}\!|^2\,\label{p01-ndiag},
	\end{eqnarray}
	where we used the property $\sum_j|\!\braket{x|\e_j}\!|^2=\|\!\ket{x}\!\|^2=1$. Let us now consider the Chapman--Kolmogorov equality~\eqref{eq:ck-xy} with $x=y$. We must have, for all $x=1,\dots,d$ and all $t\geq r\geq s\geq 0$,
	\begin{eqnarray}
		0&=&\mathbb{P}(x,t|x,s)-\sum_{z=0}^d\mathbb{P}(x,t|z,r)\mathbb{P}(z,r|x,s)\nonumber\\
		&=&\sum_{j,\ell=1}^d\e^{-\gamma_j(t-r)}\e^{-\gamma_\ell(r-s)}|\braket{x|\e_j}|^2|\braket{x|\e_\ell}|^2\left[\delta_{j\ell}-\sum_{z=1}^d|\braket{z|\e_j}|^2|\braket{z|\e_\ell}|^2\right];
	\end{eqnarray}
	that is, defining $\tau:=t-s$ and $\sigma:=r-s$, we must have\small
	\begin{equation}
		\forall\tau,\sigma\geq0,\;\forall x=1,\dots,d,\quad\sum_{j=1}^d\e^{-\gamma_j\tau}|\braket{x|\e_j}|^2\sum_{\ell=1}^d\e^{-\gamma_\ell\sigma}|\braket{x|\e_\ell}|^2\left[\delta_{j\ell}-\sum_{z=1}^d|\braket{z|\e_j}|^2|\braket{z|\e_\ell}|^2\right]=0;
	\end{equation}\normalsize
	but, since all $\gamma_j$s are distinct and the exponential functions are linearly independent, the equality above holds if and only if\small
	\begin{equation}
		\forall x,j=1,\dots,d,\quad \braket{x|\e_j}=0\;\;\text{or}\;\;\forall\sigma\geq0,\;\sum_{\ell=1}^d\e^{-\gamma_\ell\sigma}|\braket{x|\e_\ell}|^2\left[\delta_{j\ell}-\sum_{z=1}^d|\braket{z|\e_j}|^2|\braket{z|\e_\ell}|^2\right]=0,
	\end{equation}\normalsize
	which, for the same reason, is in turn equivalent to
	\begin{equation}
		\forall x,j,\ell=1,\dots,d,\quad\braket{x|\e_j}=0\;\;\text{or}\;\;\braket{x|\e_\ell}=0\;\;\text{or}\;\;\sum_{z=1}^d|\braket{z|\e_j}|^2|\braket{z|\e_\ell}|^2=\delta_{j\ell}
	\end{equation}
	and, in particular,
	\begin{equation}
		\forall x,j=1,\dots,d,\quad\braket{x|\e_j}=0\;\;\text{or}\;\;\sum_{z=1}^d|\braket{z|\e_j}|^4=1.
	\end{equation}
	Clearly, given $j=1,\dots,d$, there must exist some $x=1,\dots,d$ such that $\braket{x|\e_j}\neq0$. Consequently,
	\begin{equation}
		\sum_{z=1}^d|\braket{z|\e_j}|^4=1.
	\end{equation}
	Now, notice that $\sum_{z=1}^d|\braket{z|\e_j}|^2=\|\ket{\e_j}\|^2=1$. Consequently, the only case in which the equality above can hold is when all terms $|\braket{z|\e_j}|^2$ vanish expect a single one equaling one---that is, when the two bases coincide up to phase shifts; henceforth the claim.
	
	\section{Proof of Theorem~\ref{THM}}\label{app:2}
	
	Consider a projective measurements w.r.t. to $|x_\alpha\>$ $(\alpha = 0,1,\ldots,d$) such that	
	\begin{equation}
		|x_\alpha\> = \mu_\alpha |0\> \oplus |\tilde{x}_\alpha\> , \ \ \  |\tilde{x}_\alpha\> \in \mathcal{H}_\e ,
	\end{equation}
	and $|\mu_\alpha|^2 + \| |\tilde{x}_\alpha\>\|^2=1$. Taking $|\Psi_0\> = |\psi_0\> \otimes |\vac\>$, with $|\psi_0\> = \alpha |0\> \oplus |\psi_\e\>$ and using~\eqref{Psi-t} one finds	
	\begin{equation}\label{Psi-t2}
		|x_1\>\!\<x_1| \otimes \oper_B U_{t_1}|\Psi_0\> = |x_1\> \otimes \left( \Big[ \alpha \mu_1^*  + \< \psi_\e|\A(t_1)|\psi_\e\> \Big] |\vac\>  + \mu_1^*   \sum_{j=1}^d  b_j^\dagger(\xi_j(t_1))|\vac\> \right)   .
	\end{equation}
	Note that $|\Psi^{x_1}_{t_1}\>$ belongs now to the 2-excitation sector. Using the following property:
	\begin{equation}
		U_{\Dt_2} \oper_S \otimes b^\dagger_j(\xi_j(t_1)) U^\dagger_{\Dt_2} = \oper_S \otimes b^\dagger_j(\xi_j(\Dt_2,t_1)) ,
	\end{equation}
	one easily proves the following equalities:
	\begin{eqnarray}
		U_{\Dt_2} |0\> \otimes  |\vac\>  &=& |0\> \otimes  |\vac\> , \\
		U_{\Dt_2} |\tilde{x}_1\> \otimes  |\vac\> &=&  \A(\Dt_2) |\tilde{x}_1\> \otimes  |\vac\> + |0\> \otimes \sum_{k=1}^d b_k^\dagger(\tilde{\eta}^{(1)}_k(\Dt_2)) |\vac\> , \\
		U_{\Dt_2} |0\> \otimes  b^\dagger_j(\xi_l(t_1))|\vac\>  &=& |0\> \otimes b^\dagger_j(\xi_j(\Dt_2,t_1))|\vac\> ,\\
		U_{\Dt_2} |\tilde{x}_1\> \otimes  b^\dagger_j(\xi_j(t_1))|\vac\>  &=&  \A(\Dt_2) |\tilde{x}_1\> \otimes b^\dagger_j(\xi_{j}(\Dt_2,t_1))|\vac\> \nonumber \\
		&+& |0\> \otimes b^\dagger_j(\xi_{j}(\Dt_2,t_1)) \, \sum_{k=1}^d b^\dagger_k(\tilde{\eta}^{(1)}_{k}(\Dt_2))|\vac\> .
	\end{eqnarray}
	
	\begin{lemma}
		The following property holds:		
		\begin{equation}\label{bb}
			\left\|  b^\dagger_j(\xi_{j}(\Dt_2,t_1) \sum_{k=1}^d b^\dagger_k(\tilde{\eta}^{(1)}_{k}(\Dt_2))|\vac\> \right\|^2 =  \|\xi_{j}(t_1)\|^2 \, \sum_{k=1}^d \left\| \tilde{\eta}^{(1)}_{k}(\Dt_2)\right\|^2 .
		\end{equation}
	\end{lemma}
	\begin{proof}
		One has	
		\begin{equation}
			\left\| U_{\Dt_2} |\tilde{x}_1\> \otimes  b^\dagger_j(\xi_j(t_1))|\vac\> \right\| = \left\| |\tilde{x}_1\> \otimes  b^\dagger_j(\xi_j(t_1))|\vac\>  \right\| ,
		\end{equation}
		and hence	
		\begin{eqnarray}
			\left\| |\tilde{x}_1\> \otimes  b^\dagger_j(\xi_j(t_1))|\vac\>  \right\|^2 & = & \Big\|   \A(\Dt_2) |\tilde{x}_1\> \otimes b^\dagger_j(\xi_{j}(\Dt_2,t_1))|\vac\> \nonumber \\
			&& +|0\> \otimes b^\dagger_j(\xi_{j}(\Dt_2,t_1)) \, \sum_{k=1}^d b^\dagger_k(\tilde{\eta}^{(1)}_{k}(\Dt_2))|\vac\> \Big\|^2  \\ &=& \Big\|   \A(\Dt_2) |\tilde{x}_1\> \otimes b^\dagger_j(\xi_{j}(\Dt_2,t_1))|\vac\> \Big\|^2  \nonumber \\ &+& \left\| b^\dagger_j(\xi_{j}(\Dt_2,t_1)) \, \sum_{k=1}^d b^\dagger_k(\tilde{\eta}^{(1)}_{k}(\Dt_2))|\vac\> \right\|^2 .
		\end{eqnarray}
		Finally, using the normalization condition	
		\begin{equation}
			\| \A(\Dt_2) |\tilde{x}_1\> \|^2 + \sum_{k=1}^d \| \tilde{\eta}^{(1)}_{k}(\Dt_2)) \|^2 = \| |\tilde{x}_1\> \|^2  ,
		\end{equation}
		one proves Eq.~\eqref{bb}.
	\end{proof}
	Hence,	
	\begin{equation}
		U_{\Dt_2} \Psi^{x_1}_{t_1}  = U_{\Dt_2} |x_1\> \otimes \Big(  \alpha_1(t_1) |\vac\>  + \mu_1^*   \sum_{j=1}^d  b_j^\dagger(\xi_j(t_1))|\vac\> \Big)   ,
	\end{equation}
	with	
	\begin{equation}
		\alpha_1(t_1) = \alpha \mu_1^*  + \< \tilde{x}_1|\A(t_1)|\psi_\e\> ,
	\end{equation}
	can be represented as follows:	
	\begin{eqnarray}
		U_{\Dt_2} \Psi^{x_1}_{t_1}  &=& \mu_1 |0\> \otimes \Big( \alpha_1(t_1)  |\vac\> + \mu_1^*  \sum_{j=1}^d  b_j^\dagger(\xi_j(\Dt_2,t_1))|\vac\>  \Big)   \nonumber \\
		&+&  \alpha_1(t_1) \Big( \A(\Dt_2) |\tilde{x}_1\> \otimes  |\vac\> + |0\> \otimes \sum_{k=1}^d b_k^\dagger(\tilde{\eta}^{(1)}_k(\Dt_2)) |\vac\> \Big) \nonumber\\
		&+& \mu_1^* \Big( \A(\Dt_2) |\tilde{x}_1\> \otimes b^\dagger_j(\xi_{j}(\Dt_2,t_1))|\vac\>\!+\!|0\>\!\otimes\! b^\dagger_j(\xi_{j}(\Dt_2,t_1)) \, \sum_{k=1}^d b^\dagger_k(\tilde{\eta}^{(1)}_{k}(\Dt_2))|\vac\> \!\Big). \nonumber\\
	\end{eqnarray}	
	
	\begin{lemma}
		In the limit of flat form factors, then		
		\begin{equation}\label{A15}
			\< \vac |  b_j(\xi^*_{j}(\Dt_2,t_1)) b^\dagger_k(\tilde{\eta}^{(1)}_{k}(\Dt_2))|\vac\>  =0 ,
		\end{equation}
		for arbitrary $\Dt_2 \geq 0 $.
	\end{lemma}
	\begin{proof}
		One has	
		\begin{eqnarray}
			&&  \< \vac |  b_j(\xi^*_{j}(\Dt_2,t_1)) b^\dagger_k \tilde{\eta}^{(1)}_{k}(\Dt_2)) |\vac\> \\
			&& =  f^*_j f_k \int \mathrm{d}\omega \int_0^{t_1} \mathrm{d}s \int_0^{\Dt_2} \mathrm{d}s' \e^{\i\omega \Dt_2}  \e^{\i\omega(t_1-s)} \e^{-\i\omega(\Dt_2 -s')} \< \psi_e|\A^\dagger(s)|\e_j\>\!\<\e_k|\A(s')|\tilde{x}_1\> , \nonumber
		\end{eqnarray}
		and using	
		$$   \int \mathrm{d}\omega \;\e^{\i\omega(t_1-s + s')}  = 2\pi \delta(s-[t_1+s']) , $$
		which implies that $s=t_1+s'$, whence $s$ is outside the domain of integration $[0,t_1]$ thus implying Eq.~\eqref{A15}.
	\end{proof}
	
	Finally, using the above lemma and the property~\eqref{bb}, simple algebra leads to	
	\begin{equation}
		{\rm Tr}_B U_{\Dt_2} |\Psi^{x_1}_{t_1}\> \< \Psi^{x_1}_{t_1}| U^\dagger_{\Dt_2} = \Lambda_{\Dt_2}(|x_1\>\!\<x_1|) \mathbb{P}_1(x_1,t_1) ,
	\end{equation}
	where	
	\begin{equation}
		\mathbb{P}_1(x_1,t_1) = |  \alpha \mu_1^*  + \< \tilde{x}_1|\A(t_1)|\psi_\e\> |^2  +  |\mu_1|^2   \sum_{j=1}^d  \|\xi_j(t_1) \|^2   .
	\end{equation}
	It should be clear from the above analysis that this scheme may be immediately generalized for $n$ measurement scenario provided that Prop.~\ref{PRO-bbb} holds.
	
	\section{Proof of Prop.~\ref{PRO-bbb}}	
	We shall prove it by induction: assuming that Eq.~\eqref{bbb} holds for $(n-1)$ measurements, we will show that it holds for $n$ measurements as well. One finds	
	\begin{eqnarray}
		&& U_{\Dt_{n+1}} |\tilde{x}_{n+1}\> \otimes  b^\dagger_{j_1}(\xi_{j_1}(t_n-t_1,\Dt_1)) \, \prod_{k=2}^n
		\sum_{j_k=1}^d b^\dagger_{j_k}(\tilde{\eta}_{j_k}^{(k-1)}(t_n-t_k,\Dt_k)) |\vac\>
		\nonumber \\  && =
		\A(\Dt_{n+1}) |\tilde{x}_{n+1}\> \otimes  b^\dagger_{j_1}(\xi_{j_1}(t_{n+1}-t_1,\Dt_1)) \, \prod_{k=2}^n
		\sum_{j_k=1}^d b^\dagger_{j_k}(\tilde{\eta}_{j_k}^{(k-1)}(t_{n+1}-t_k,\Dt_k)) |\vac\>  \nonumber \\ && + |0\> \otimes b^\dagger_{j_1}(\xi_{j_1}(t_{n+1}-t_1,\Dt_1)) \, \prod_{k=2}^{n+1}
		\sum_{j_k=1}^d b^\dagger_{j_k}(\tilde{\eta}_{j_k}^{(k-1)}(t_{n+1}-t_k,\Dt_k)) |\vac\> .
	\end{eqnarray}
	Hence, due to the fact that $U_{\Dt_{n+1}}$ is unitary,	
	\begin{eqnarray}
		&& \|\tilde{x}_{n+1}\>\|^2 \left\| b^\dagger_{j_1}(\xi_{j_1}(t_n-t_1,\Dt_1)) \, \prod_{k=2}^n
		\sum_{j_k=1}^d b^\dagger_{j_k}(\tilde{\eta}_{j_k}^{(k-1)}(t_n-t_k,\Dt_k)) |\vac\> \right\|^2 \nonumber \\
		&& =\,  \left\|\A(\Dt_{n+1}) |\tilde{x}_{n+1}\> \right\|^2 \, \left\|  b^\dagger_{j_1}(\xi_{j_1}(t_{n+1}-t_1,\Dt_1)) \, \prod_{k=2}^n
		\sum_{j_k=1}^d b^\dagger_{j_k}(\tilde{\eta}_{j_k}^{(k-1)}(t_{n+1}-t_k,\Dt_k)) |\vac\> \right\|^2 \nonumber\\
		&& +\,  \left\| b^\dagger_{j_1}(\xi_{j_1}(t_{n+1}-t_1,\Dt_1)) \, \prod_{k=2}^{n+1}
		\sum_{j_k=1}^d b^\dagger_{j_k}(\tilde{\eta}_{j_k}^{(k-1)}(t_{n+1}-t_k,\Dt_k)) |\vac\> \right\|^2 .
	\end{eqnarray}
	Now, using Eq.~\eqref{bbb}, i.e.	
	\begin{eqnarray}
		&& \left\|  b^\dagger_{j_1}(\xi_{j_1}(t_{n+1}-t_1,\Dt_1)) \, \prod_{k=2}^n
		\sum_{j_k=1}^d b^\dagger_{j_k}(\tilde{\eta}_{j_k}^{(k-1)}(t_{n+1}-t_k,\Dt_k)) |\vac\> \right\|^2  \\
		&& = \left\|  b^\dagger_{j_1}(\xi_{j_1}(\Dt_1))\right\|^2  \, \prod_{k=2}^n \left\|
		\sum_{j_k=1}^d b^\dagger_{j_k}(\tilde{\eta}_{j_k}^{(k-1)}(t_{n+1}-t_k,\Dt_k)) \right\|^2  , \nonumber
	\end{eqnarray}
	one obtains	
	\begin{eqnarray}
		&& \left\|  b^\dagger_{j_1}(\xi_{j_1}(\Dt_1)) \, \prod_{k=2}^{n+1}
		\sum_{j_k=1}^d b^\dagger_{j_k}(\tilde{\eta}_{j_k}^{(k-1)}(t_{n+1}-t_k,\Dt_k))|\vac\> \right\|^2  \\
		&& = \left(\left\| |\tilde{x}_{n+1}\> \right\|^2 - \left\| \A(\Dt_{n+1}|\tilde{x}_{n+1}\> \right\|^2\right)\!\|  b^\dagger_{j_1}(\xi_{j_1}(\Dt_1))\|^2  \prod_{k=2}^{n} \left\|
		\sum_{j_k=1}^d b^\dagger_{j_k}(\tilde{\eta}_{j_k}^{(k-1)}(t_{n+1}-t_k,\Dt_k)) \right\|^2 .  \nonumber
	\end{eqnarray}
	Finally, using the normalization property	
	\begin{equation}
		\left\| \A(\Dt_{n+1})|\tilde{x}_{n+1}\> \right\|^2 + \sum_{j=1}^d \left\| \tilde{\eta}^{(n)}_j(\Dt_{n+1})\right\|^2 = \| |\tilde{x}_{n+1}\> \|^2 ,
	\end{equation}
	one proves	
	\begin{eqnarray}
		&&\left\|  b^\dagger_{j_1}(\xi_{j_1}(\Dt_1))  \, \prod_{k=2}^{n+1}
		\sum_{j_k=1}^d b^\dagger_{j_k}(\tilde{\eta}_{j_k}^{(k-1)}(t_{n+1}-t_k,\Dt_k))|\vac\> \right\|^2  \nonumber \\ && = \|  \xi_{j_1}(\Dt_1)\|^2  \, \prod_{k=2}^{n+1}
		\left\|
		\sum_{j_k=1}^d \tilde{\eta}_{j_k}^{(k-1)}(t_{n+1}-t_k,\Dt_k)\right\|^2  .
	\end{eqnarray}

\end{document}